\documentclass[12pt]{article}
\pdfoutput=1
\usepackage[utf8]{inputenc} 
\usepackage[T1]{fontenc} 
\usepackage[english]{babel}
\usepackage[margin=1.25in]{geometry}
\usepackage{amsmath}
\usepackage{amssymb, amsthm, mathrsfs}
\usepackage{mathtools, nccmath}
\usepackage{bbm}
\usepackage{hyperref}
\usepackage{pdfpages}
\usepackage{graphicx}
\usepackage{subfig}
\usepackage{caption}
\usepackage{verbatim}
\usepackage{multirow}
\usepackage{rotating}
\usepackage{tabularx}
\usepackage{xcolor}
\usepackage{array}
\usepackage{makecell}
\usepackage{setspace}
\usepackage{natbib}
\usepackage[capitalize]{cleveref}
\usepackage{enumerate}
\usepackage{thm-restate}
\usepackage{csquotes}
\usepackage[capposition=top]{floatrow}
\usepackage{soul}
\usepackage{tikz}
\usetikzlibrary{calc,decorations}

\usepackage[nocomma]{optidef}

\graphicspath{ {./images/} }

\newcommand{\N}{\mathbb{N}}

\newcommand{\cP}{\mathcal{P}}

\theoremstyle{definition}
\newtheorem{definition}{Definition}

\newtheorem{assumption}{Assumption}

\theoremstyle{plain}
\newtheorem{theorem}{Theorem}
\newtheorem{corollary}{Corollary}
\newtheorem{lemma}{Lemma}
\newtheorem{proposition}{Proposition}

\theoremstyle{remark}

\usepackage{color-edits}
\addauthor{yl}{blue}    
\addauthor{xq}{brown}    

\providecommand{\keywords}[1]
{
  \small	
  \textbf{\textit{Keywords---}} #1
}
\providecommand{\JEL}[1]
{
  \small	
  \textbf{\textit{JEL---}} #1
}

\title{Allocating Resources under Strategic Misrepresentation\thanks{We thank Asher Wolinsky, Bruno Strulovici, and Wojciech Olszewski for advice, helpful conversations, and comments. We also thank Ian Ball, Dirk Bergemann, Alex Bloedel, Simon Board, Eddie Dekel, Piotr Dworczak, Jeff Ely, Kira Goldner, Carl-Christian Groh, Yingni Guo, Marina Halac, Andrei Iakovlev, Annie Liang, Bart Lipman, Brendan Lucier, Deniz Kattwinkel, Joshua Mollner, Kyohei Okumura, Alessandro Pavan, Marcin Pęski,  Abhishek Sarkar, James Schummer, Christopher Snyder, Philipp Strack, Matthew Thomas, and participants at the UCLA Theory Seminar, 34th Stony Brook International Conference on Game Theory, Australasian Economic Theory Week, NUS Theory Lunch for helpful suggestions, comments, and discussions. Yingkai Li acknowledges financial support from NUS Start-up Grant. Part of the work was done while Yingkai Li was a Postdoc at Yale University under the support of the Sloan Research Fellowship, grant no.~FG-2019-12378.}}
\author{Yingkai Li\thanks{Department of Economics, National University of Singapore.
Email: \texttt{yk.li@nus.edu.sg}}
\and
Xiaoyun Qiu\thanks{Department of Economics, Dartmouth College.
Email:~\texttt{xiaoyun.qiu@dartmouth.edu}}}

\date{}

\begin{document}
\onehalfspacing
\newcommand{\setsize}[1]{{\left|#1\right|}}

\newcommand{\floor}[1]{
{\lfloor {#1} \rfloor}
}
\newcommand{\bigfloor}[1]{
{\left\lfloor {#1} \right\rfloor}
}


%
%
\newcommand{\given}{\,|\,}
\newcommand{\wgiven}{\,\mid\,}

\newcommand{\prob}[2][]{\text{\bf Pr}\ifthenelse{\not\equal{}{#1}}{_{#1}}{}\!\left[{\def\givenn{\middle|}#2}\right]}
\newcommand{\expect}[2][]{\text{\bf E}\ifthenelse{\not\equal{}{#1}}{_{#1}}{}\!\left[{\def\givenn{\middle|}#2}\right]}

\newcommand{\tparen}{\big}
\newcommand{\tprob}[2][]{\text{\bf Pr}\ifthenelse{\not\equal{}{#1}}{_{#1}}{}\tparen[{\def\given{\tparen|}#2}\tparen]}
\newcommand{\texpect}[2][]{\text{\bf E}\ifthenelse{\not\equal{}{#1}}{_{#1}}{}\tparen[{\def\given{\tparen|}#2}\tparen]}

\newcommand{\sprob}[2][]{\text{\bf Pr}\ifthenelse{\not\equal{}{#1}}{_{#1}}{}[#2]}
\newcommand{\sexpect}[2][]{\text{\bf E}\ifthenelse{\not\equal{}{#1}}{_{#1}}{}[#2]}

\newcommand{\lbr}[1]{\left\{#1\right\}}
\newcommand{\rbr}[1]{\left(#1\right)}
\newcommand{\cbr}[1]{\left[#1\right]}

\newcommand{\suchthat}{\,:\,}

\newcommand{\partialx}[2][]{{\tfrac{\partial #1}{\partial #2}}}
\newcommand{\nicepartialx}[2][]{{\nicefrac{\partial #1}{\partial #2}}}
\newcommand{\dd}{{\,\mathrm d}}
\newcommand{\ddx}[2][]{{\tfrac{\dd #1}{\dd #2}}}
\newcommand{\niceddx}[2][]{{\nicefrac{\dd #1}{\dd #2}}}
\newcommand{\grad}{\nabla}

\newcommand{\symdiff}{\triangle}
\newcommand{\abs}[1]{\left|#1\right|}
\newcommand{\indicate}[1]{{\bf 1}\left[#1\right]}
\newcommand{\reals}{\mathbb{R}}
\newcommand{\posreals}{\reals_+}
\newcommand{\naturals}{\mathbb{N}}
\newcommand{\posnaturals}{\naturals_+}
\newcommand{\supp}{\text{supp}}

\newcommand{\alloc}{Q}
\newcommand{\allocs}{\boldsymbol{Q}}
\newcommand{\signalalloc}{x}
\newcommand{\signalallocs}{\boldsymbol{x}}
\newcommand{\mechalloc}{y}
\newcommand{\mechallocs}{\boldsymbol{y}}
\newcommand{\contestalloc}{q}
\newcommand{\contestallocs}{\boldsymbol{q}}
\newcommand{\efficientalloc}{\alloc_{\rm E}}
\newcommand{\efficientutil}{\util_{\rm E}}
\newcommand{\optimalalloc}{\alloc_{\alpha}}
\newcommand{\optimalutil}{\util_{\alpha}}
\newcommand{\allocspace}{X}

\newcommand{\intalloc}{\mathcal{Q}}

\newcommand{\type}{\theta}
\newcommand{\types}{\boldsymbol{\theta}}
\newcommand{\typespace}{\Theta}
\newcommand{\Typespace}{\boldsymbol{\typespace}}
\newcommand{\reporttype}{\hat{\theta}}
\newcommand{\reporttypes}{\boldsymbol{\hat{\theta}}}
\newcommand{\signal}{s}
\newcommand{\signals}{\boldsymbol{\signal}}
\newcommand{\signalspace}{S}
\newcommand{\Signalspace}{\boldsymbol{\signalspace}}
\newcommand{\reportsignal}{\hat{s}}
\newcommand{\reportsignals}{\boldsymbol{\reportsignal}}
\newcommand{\generalsignal}{\tilde{s}}
\newcommand{\generalsignals}{\boldsymbol{\generalsignal}}

\newcommand{\effort}{e}
\newcommand{\util}{U}
\newcommand{\utils}{\boldsymbol{U}}
\newcommand{\expostutil}{u}
\newcommand{\expostutils}{\boldsymbol{u}}
\newcommand{\mech}{\mathcal{M}}
\newcommand{\strategy}{\sigma}
\newcommand{\dist}{F}
\newcommand{\dists}{\boldsymbol{F}}
\newcommand{\density}{f}
\newcommand{\hazard}{\frac{f}{1-F}}
\newcommand{\cost}{c}
\newcommand{\Cost}{C}
\newcommand{\prize}{v}
\newcommand{\prizes}{\boldsymbol{\prize}}
\newcommand{\quant}{q}

\newcommand{\dista}{D}
\newcommand{\utild}{\tilde{\util}}

\newcommand{\signalrecommend}{\tilde{\signal}}
\newcommand{\signalrecommends}{\tilde{\signals}}
\newcommand{\ability}{\eta}
\newcommand{\obj}{{\rm Obj}}
\newcommand{\primed}{^\dagger}
\newcommand{\dprimed}{^\ddagger}
\newcommand{\contest}{\mathcal{T}}
\newcommand{\instance}{\mathcal{I}}
\newcommand{\opt}{{\rm OPT}}

\newcommand{\efficientallocn}[1]{\alloc_{{\rm E},#1}}
\newcommand{\optimalallocn}[1]{\alloc_{\alpha,#1}}
\newcommand{\efficientutiln}[1]{\util_{{\rm E},#1}}
\newcommand{\optimalutiln}[1]{\util_{\alpha,#1}}
\newcommand{\ranking}{r}
\newcommand{\tie}{z}
\newcommand{\ce}{\Sigma}

\maketitle

\begin{abstract}
We study how to allocate resources to participants who can strategically misrepresent their deservingness at a cost. A principal assigns item(s) (or money) among multiple agents on the basis of their costly signals. Each agent's signal reflects their private type in the absence of misrepresentation but can be inflated above their true type at a cost. The principal is a social planner who aims to maximize the weighted average of \emph{matching efficiency} and a \emph{utilitarian objective}. Strategic misrepresentation introduces novel incentive-compatibility constraints, under which we characterize the optimal mechanism. We apply our characterization to two kinds of markets, distinguished by resource scarcity, and show that the principal strictly benefits from randomizing the allocations based on costly signals when the population of participants is large enough. Interestingly, in large markets with scarce resources, the format of the optimal mechanism converges to a winner-takes-all contest; however, there is a non-diminishing value in randomizing allocations to \emph{middle} types as the population of participants grows.
\end{abstract}


\JEL{D47, D61, D82}

\section{Introduction} \label{sec:intro}
In many real-world scenarios, the designer aims to allocate scarce public resources to the most deserving agents to maximize social value. Ideally, in an environment with strong institutions and effective oversight, resources would be allocated based on verifiable evidence of deservingness, which is essentially unobservable characteristics made manifest through rigorous, truthful reporting. However, in settings characterized by weak oversight, the signals used to identify merit become malleable. Instead of providing truthful evidence, agents can exert costly strategic effort to misrepresent their underlying characteristics.
This creates a fundamental tension: while these signals convey information, they are susceptible to inflation. Agents are incentivized to invest in socially wasteful ways to inflate their signals without improving their underlying merit. Consequently, the designer faces a dual challenge: they must distinguish between genuine deservingness and manufactured signals, knowing that the ``evidence'' provided is a function of both an agent's true type and their strategic effort.\footnote{Given that the public resources could be money itself, it is usually not effective to use monetary transfers to provide incentives to solicit information.}

Examples of such real-world scenarios include the following: In public procurement, rather than the Request for Proposal (RFP) capturing a firm's true capacity (truthful evidence), weak oversight allows firms to manufacture ``paper capacity'' via shell companies (costly inflation)\citep{sba2025, usdoj2023}.
In public housing, rather than a transparent audit of financial need, the system relies on self-reporting that the designer cannot perfectly verify, leading to strategic underreporting.\footnote{For example, \citet{gao2005} reports that applicants underreported income and conceal assets to meet eligibility thresholds for public housing.}
Similarly, in scientific funding, peer review is intended to be a mechanism for truthful evidence of promise, but in a ``publish or perish'' environment with high stakes and limited verification of long-term impact, it devolves into the costly inflation of research claims \citep{NIHgrant}.

Although these applications are reminiscent of the classic signaling framework \citep[e.g.,][]{spence1973job}, they differ in two major ways. First, the signaling technology is different. In the classic framework, effort is the only way to make the underlying characteristics visible. In our setting, characteristics would be observable and verifiable under perfect oversight; however, under the weak oversight we model, agents gain the ability to inflate these signals through private effort. We explicitly model such strategic misrepresentation behavior: agents may produce inflated signals at a cost, while downward misreporting is costless. Second,  the classic signaling framework considers a decentralized setting where agents are paid a fair wage based on their signal, whereas we consider a centralized setting where agents compete for limited resources from a centralized planner. These two departures interact with each other and allow us to uncover interesting economic effects from costly signaling and competition.  

Motivated by these applications, we study resource allocation under strategic misrepresentation and apply these insights to the design of public programs.  A principal aims to allocate item(s) among multiple agents. The objective consists of two parts. The first part is \emph{matching efficiency}, a reduced form notion to capture how far away an allocation rule is from assortative matching. Assortative matching can be viewed as an achievement-based meritocracy, where merit is defined purely as the realized signal, regardless of how it was produced. Under this purely output-oriented view, the ``best'' (highest signal) deserves the item because they demonstrate the highest performance. The second part is a \emph{utilitarian objective} \`{a} la Harsanyi, the sum of the equally weighted utilities of the agents. Since it gives equal weight to every agent's utility, the utilitarian objective captures the principal's preference for fairness.\footnote{This utilitarian objective is also fair in the procedural sense, because it uses an unbiased "lottery of birth" where everyone has the same chance to be in any position.} The principal seeks to maximize a weighted combination of these objectives, accounting for endogenous misrepresentation incentives.

We first characterize the optimal mechanisms. Incorporating strategic misrepresentation introduces novel incentive compatibility (IC) constraints under which the principal's optimization problem is not convex. The conventional extremal point approach fails. Yet, we circumvent this challenge by showing that the optimal mechanism is symmetric when agents are ex-ante identical. This greatly simplifies the problem, enabling us to apply optimal control techniques to characterize the optimal mechanism. It exhibits the following form: each agent’s type space is partitioned into disjoint intervals falling into one of three categories: (1) a \emph{no-tension} interval, (2) the \emph{no-effort} interval, and (3) the \emph{efficient} interval. Types in the no-effort interval are pooled and assigned a coarse ranking, while types in the no-tension and efficient intervals preserve their strict ordering. Allocation is then determined efficiently based on this composite ranking. In equilibrium, agents in the no-tension and no-effort intervals exert no effort, whereas types in the efficient interval engage in positive effort. Importantly, while these three categories of intervals partition the type space, each category of interval may occur multiple times, and their ordering depends on primitives such as the type distribution and the number of available items.

Intuitively, the principal seeks to allocate items efficiently whenever doing so does not incentivize costly misrepresentation. This is feasible in the \emph{no-tension} interval, where truthfully revealing one's type is incentive-compatible, making efficient allocation optimal. However, when efficient allocation induces agents to misrepresent their types despite of the effort cost, the optimal mechanism depends on the principal's objective. If the weight on agents' utilities is sufficiently high, it may still be optimal to allocate efficiently while inducing effort (as in the \emph{efficient} interval).
Alternatively, the principal may randomize allocations to flatten the marginal benefit of effort, thereby eliminating incentives for misrepresentation. This occurs in the \emph{no-effort} interval, where all types in the interval exert no effort. A key insight from our analysis is that these extreme treatments, either fully suppressing effort or allocating efficiently, are optimal. Intermediate designs that induce small but positive signaling costs are never optimal.

To further illustrate the optimal mechanism, consider a simple setting where there are two agents, one item, and the distribution of each agent's type is identical and convex (See \Cref{fig:implement_convex_efficient} for a more detailed explanation). In this example, there is exactly one interval for each of the three regions, where the no-tension interval is followed by the no-effort interval, which is followed by the efficient interval. Hence, each agent's type space is divided into three categories: low, middle, and high. The notable feature of the optimal mechanism is that when both agents' types fall into the middle range, the item is allocated randomly, but the agent with a higher type has a higher probability (which is strictly less than $1$) of getting the item. When at least one agent's type falls outside the middle range, the item is allocated to the agent with a higher signal. Randomizing the item allocation is optimal among the middle types. This is because (1) it directly eliminates the incentive to exert signaling effort from middle types, and (2) it indirectly lowers the signaling effort required for high types. The total reduction in effort, and thus improvement in utilities, exceeds the loss in matching efficiency by deviating from assortative matching.

We apply this characterization to study large markets, distinguished by the scarcity of resources, which will inform us about the design of public programs.
In the first case, the principal allocates a single item,\footnote{The analysis extends naturally to any constant number of items and the allocation of money.} reflecting settings such as contracting for procurement or research funding, where the number of winners is small relative to the pool of applicants. As the number of agents grows, the optimal mechanism converges to a winner-takes-all (WTA) contest in structure: an efficient interval expands to cover the entire type space. However, the principal’s expected payoff under the WTA contest does not converge to the optimal payoff in the large-market limit.
The intuition is that, for any sufficiently large but finite number of agents, the WTA contest puts excessive pressure on agents with types close to the highest in the support in order to win the item, leading to significant costly efforts from those types. In contrast, by introducing a small but non-empty no-effort interval in the optimal mechanism for these high types, which randomizes the allocation, the principal can significantly reduce the cost of signaling with only a negligible loss in efficiency, as all these high types are almost equally qualified for receiving the item. Our finding contrasts with results in the large-contest literature that use a continuum of agents to approximate a finite market \citep[e.g.,][]{OS2016large}, where WTA contests often appear optimal.

In the second case, the number of items grows proportionally with the number of agents. This reflects programs such as college admissions or government benefit programs like public housing or food subsidies, where a sizable share of participants receive an item.
In this case, if the items were allocated efficiently, all agents with types above a certain cutoff would receive an item.
However, we find that in the optimal mechanism, the principal randomizes the allocation for types around the cutoff ( the ``middle'' types) to eliminate their incentives to exert costly effort. This increases the expected utilities of both the ``middle'' types and slightly higher types (those whose types are above the ``middle'' types, but not significantly above the cutoff), at the cost of only slightly reducing matching efficiency for the ``middle'' types.
Our result resonates with the Director's Law, which posits that public programs often disproportionately benefit the middle class. It also aligns with empirical findings by \citet{KLOST2022pareto}, who show that randomly assigning college seats to lower-scoring students in Turkey reduces stress across the board.\footnote{In our model, “low” types correspond to students with no realistic admission prospects. Randomizing among middle types in our setting mirrors the random allocation to low-scoring students in their empirical context.}

\subsection{Related Work}

Our paper is closely related to the literature on costly signaling and money burning \citep[e.g.,][]{chakravarty2006manna,hartline2008,condorelli2012money,finkelstein2019take,akbarpour2024redistributive,yang2024comparison}.
We have two main departures from this literature: the designer's objective and the agents' misrepresentation technology.
First, in our paper,  the principal seeks to maximize the weighted averages of matching efficiency (a preference over allocations) and the agents' utilities. If the weight on matching efficiency is zero, the principal's objective is purely from redistributive concerns \citep[e.g.,][]{dworczak2021redistribution,akbarpour2024redistributive}.
Second, prior work in this literature imposes a restrictive quasi-linear structure on agents' preferences, while our work relaxes this restriction. Cases that have been studied in the literature include  (1) agents' utilities are linear in transfers \citet{hartline2008}, (2) agents' utilities are linear in transfers \citep{akbarpour2024redistributive,yang2024comparison}, and (3) the designer can choose the cost function, essentially resulting in a  
linear utility function \citep{condorelli2012money}. 
In contrast, our paper considers an exogenous, type-dependent cost structure that is more realistic and does not follow a linear structure. As shown in \citet{li2025mechanism}, the standard payoff equivalence fails in our environment, and the implementation of allocation rules and the (non-)coordination of costly signaling choices play a crucial role in reducing the cost of screening. 
Our paper provides a characterization of the optimal mechanism under this more realistic cost structure and applies these theoretical insights to the design of public programs.

Additionally, our analysis connects mechanism design under costly signaling to contests \citet{li2025mechanism}. The non-coordination mechanism is essentially an all-pay contest. Our characterization of optimal contests in large markets sharply contrasts with the results in large contests by \citet{olszewski2020performance}, as our paper provides an environment where the optimal payoff in a finite large contest cannot be approximated by the limiting case.

Our paper is related to the literature on falsification and costly lying \citep[e.g.,][]{green1986partially,hardt2016strategic,perez2022test,perez2023fraud,perez2024score}. 
The paper most closely related to ours is the contemporaneous work by \citet{perez2024score}. They focus on a single-agent model and a different design objective.  In contrast, we focus on a multi-agent environment and  the principal's main objective is to maximize the weighted averages of matching efficiency and the agents' utilities. As a result, the competition effect is absent in \citet{perez2024score} but is captured by our model. This allows us to study competition for limited resources.  

The signals in our model can be viewed as messages sent by the agents to the principal, and the corresponding effort costs can be viewed as lying costs. This ties our model to the literature on partial verification, where, however, lying costs are assumed to be either $0$ or $\infty$, depending on the agent's true type and his message \citep[e.g.,][]{green1986partially}.
In our paper, we allow for a more general form of lying cost, where the cost of lying can increase with the magnitude of the lie. 
Furthermore, in contrast to the literature on evidence \citep[e.g.,][]{BDL2014,mylovanov2017optimal}, our model involves evidence that is not ``hard,'' since an agent can (at a cost) fabricate a signal different from his true type.

Our paper is also related to the literature on grant making\citep[e.g.,][]{carnehl2025designing,adda2024grantmaking}. \citet{adda2024grantmaking} study a setting similar to our large market model with abundant resources. They study a general equilibrium model where the competition effect is absent. Thus, resources are allocated by a cut-off rule. In contrast, we adopt a mechanism design approach that incorporates the competition effect in grant allocation. Due to the non-negligible competition effect, our optimal mechanism randomizes allocation near the cut-off.

The costly signaling aspect of our paper resembles the literature on signaling \citep[e.g.,][]{spence1973job} and gaming \citep[e.g.,][]{FK2019muddled,Ball2019}, which studies manipulative behaviors in signaling games. We depart to these lines of literature in two ways. The first distinction is that in signaling games, there is a competitive market that pays each agent a wage corresponding to their estimated type, whereas we adopt a mechanism design perspective in these markets and characterize the optimal mechanisms for allocating resources using costly signals as screening devices. The second distinction is that we study a different signaling technology that captures the commonly observed strategic misrepresentation behaviors in practice.

\section{Model}
\label{sec:model}
The principal has $k\geq 1$ identical items to allocate among $n> k$ agents based on their costly signals.
Let $\allocspace \subseteq [0,1]^n$ be the space of feasible allocations where $\sum_{i=1}^n\signalalloc_i\leq k$ and $0\leq \signalalloc_i \leq 1$.

Each agent $i$ has a private type~$\type_i$ drawn independently from a publicly known distribution~$\dist_i$ supported on $\Theta_i\subseteq \reals_+$.
For agent $i$ with type $\type_i$, his cost of generating a costly signal $\signal_i\in \Signalspace=\reals_+$ is $c(\signal_i,\type_i) = \ability\cdot(\signal_i-\type_i)^+$, where $\eta$ is a parameter that captures the agent's marginal cost of inflating the signals.
For every agent, he receives a utility $1$ for receiving an item and $0$ for not receiving anything. 
Thus, his utility from obtaining an item with probability $\signalalloc_i$ and generating a costly signal $\signal_i$ is 
\begin{align*}
\expostutil_i(\signalalloc_i,\signal_i,\type_i)=\signalalloc_i - \eta\cdot(\signal_i-\type_i)^+.
\end{align*}

\paragraph{Mechanisms.} 
It is without loss of generality to focus on direct mechanisms. A direct mechanism $(\signalrecommends,\mechallocs)$ consists of a signal recommendation rule $\signalrecommends: \prod_{i=1}^n \Typespace_i\to\Signalspace^n$, where $\signalrecommends=(\signalrecommend_i)_{i=1}^n$, and an allocation rule $\mechallocs:\Signalspace^n\to \allocspace$, where $\mechallocs=(\mechalloc_i)_{i=1}^n$. 
By \citet{li2025mechanism}, it is without loss of optimality to focus on direct non-coordination mechanisms.\footnote{
\citet{li2025mechanism} shows that if the principal favors assortative matching, which is the case in our setting, then it is without loss of optimality to focus on direct non-coordination mechanisms.} 
In a direct non-coordination mechanism, the signal recommendation policy is such that $\signalrecommends_i:\Typespace_i\to \Signalspace_i$ for every $i$, and the allocation rule is a function of the observed signals.
Based on the recommendation, each agent $i$ produces a signal $\signal'_i$.
The principal observes the signal profile $\signals'=(\signal_i')_{i=1}^n$.
Each agent $i$ receives an item with probability $\mechalloc_i(\signals')$.  
For convenience, from now on, when it does not cause confusion, we will omit the term non-coordination. Thus, any direct mechanism is, by default, a direct non-coordination mechanism.

\paragraph{Interim approach.}
For analytical convenience, we adopt the following interim approach to characterize the optimal mechanism. It is equivalent to work in the space of  \emph{interim allocation} and  \emph{interim utility}  as long as they satisfy the interim feasibility constraint \citep[e.g.,][]{border1991implementation,che2013generalized}. 

Given any direct mechanism $(\signalrecommends,\mechallocs)$, the interim allocation and the interim utility of agent~$i$ with private type $\type_i$ are
\begin{align}\label{interim}
\alloc_i(\type_i)&=\expect[\types_{-i}]{\mechalloc_i(\signalrecommend_i(\type_i),\signalrecommends_{-i}(\types_{-i}))},\tag{consistency}\\
\util_i(\type_i) &=
\expect[\types_{-i}]{\expostutil_i(\mechalloc_i(\signalrecommend_i(\type_i),\signalrecommends_{-i}(\types_{-i})),\signalrecommend_i(\type_i),\type_i)}.\nonumber
\end{align}
 
A direct mechanism $(\signalrecommends,\mechallocs)$ is incentive compatible (IC) if for any agent $i$ and any types $\type_i,\type'_i$, 
his expected utility from truthfully reporting his type and following the signal recommendation is not lower than that from misreporting as $\type'_i$ and producing the recommended signal $\signalrecommend_i(\type'_i)$. That is,
\begin{align}\label{eq:IC1}
\util_i(\type_i) \geq \max_{\type'_i} \expect[\types_{-i}]{
\expostutil_i(\mechalloc_i(\signalrecommend_i(\type'_i),\signalrecommends_{-i}(\types_{-i})),\signalrecommend_i(\type'_i),\type_i)}.\tag{IC}
\end{align}
Denote the interim allocation profile and interim utility profile by $\allocs=(\alloc_i)_{i=1}^n$ and $\utils=(\util_i)_{i=1}^n$, respectively. 
We say $(\allocs,\utils)$ is implementable if there exists an IC non-coordination mechanism that implements it, i.e., the consistency constraints are satisfied. 

The principal aims to maximize a weighted average between matching efficiency and the agents' utilities. Given an allocation profile $\signalallocs=(\signalalloc_i)_{i=1}^n$, a signal profile $\signals=(\signal_i)_{i=1}^n$ and the agents' type profile $\types=(\type_i)_{i=1}^n$, the principal's payoff is 
$\alpha \cdot \sum_i \type_i\cdot \signalalloc_i 
+ (1-\alpha)\cdot\sum_i \expostutil_i(\signalalloc_i,\signal_i,\type_i)$, for any arbitrary  $\alpha\in[0,1]$.
Thus, the principal's problem is equivalent to one where she chooses $(\allocs,\utils)$ subject to \eqref{interim} and \eqref{eq:IC1} to maximize 
\begin{equation}\label{exp:principal_payoff}
\text{Obj}_{\alpha}(\allocs,\utils)= \expect[\types]{\alpha \cdot \sum_i \type_i\cdot \alloc_i(\type_i) 
+ (1-\alpha)\cdot\sum_i \util_i(\type_i)}.
\end{equation}

\paragraph{Symmetric environment.}
We assume that the density function $\density_i$ exists for all $i$. We assume that the agents are ex-ante homogeneous, i.e., $\Theta_i=\Theta=[\underline{\type},\bar{\type}]$, $\dist_i=\dist$ and $\density_i=\density$ for all $i$. 
Moreover, $\density(\type_i) >0$ for any $\type_i\in \Theta$.

\section{Optimal Mechanism}
\label{sec:optimal}
To characterize the optimal mechanism, we first provide a characterization of the incentive compatibility (IC) constraints. It is quite different from traditional settings because of the way we model strategic misrepresentation. This novel IC condition introduces non-convexity into the principal's optimization problem.
Nonetheless, we establish that the optimal mechanism is symmetric when agents are ex-ante symmetric. 
Finally, we characterize the optimal mechanism shaped by the novel IC condition.

\subsection{Incentive compatibility}
We first provide a closed-form characterization of the incentive compatibility conditions in any direct non-coordination mechanism that implements a monotone allocation rule.
Note that focusing on monotone allocations remains without loss of optimality \citep[c.f.,][]{li2025mechanism}.

\begin{lemma}[\citealp{li2025mechanism}]
\label{lmm:monotone optimal}
For any interim allocation--utility pair $(\allocs,\utils)$ that is implementable,
there exists another interim allocation--utility pair $(\allocs\primed,\utils\primed)$ with monotone allocation $\allocs\primed$ that is implementable and yields a weakly higher objective value for all $\alpha\in [0,1]$, where $\alpha$ is the welfare weight on \emph{matching efficiency} defined in \eqref{exp:principal_payoff}.
\end{lemma}

\begin{restatable}{lemma}{lemMonotoneImplement}\label{lem:monotone alloc-util implementation}
An interim allocation--utility pair $(\allocs,\utils)$ with monotone $\allocs$ is implementable by a non-coordination mechanism
if and only if $\allocs$ is interim feasible,
and for any agent $i$ with type $\type_i$,\footnote{The function $\util_i$ may not be differentiable everywhere. 
For any type $\type_i$ such that $\util_i$ is not differentiable at $\type_i$, we let $\util'_i(\type_i)$ denote any subgradient (or simply the left and right derivative) of the function $\util_i$. It is not hard to show that $\util_i$ is a continuous and monotone function and hence is differentiable almost everywhere.}
\begin{align}\label{eq:IC}
\text{(1)}\ \ \util'_i(\type_i)\in [0, \ability];
\quad \text{(2)}\ \ \util_i(\type_i)\leq \alloc_i(\type_i);
\quad \text{(3)}\ \ \util'_i(\type_i)=\ability \text{ if }
\util_i(\type_i) < \alloc_i(\type_i).
\tag{IC}
\end{align}
\end{restatable}

The idea behind \cref{lem:monotone alloc-util implementation} is as follows. For any interim allocation--utility pair $(\allocs,\utils)$ that is implementable by a non-coordination mechanism, there is a level constraint and a slope constraint on interim utility. 
The level constraint is intuitive: because the effort costs are non-negative, each agent's utility is bounded above by his allocation.
The slope constraint says that the marginal increase in the interim utility is bounded above by the marginal cost of effort; if this were not the case, then a low type would have an incentive to misreport and produce the recommended signal for a higher type. 
Finally, if the level constraint is slack at any type $\type$, 
the equilibrium effort for type~$\type$ must be strictly positive. 
In order to eliminate the incentives for higher types to deviate to~$\type$, 
the slope constraint must be binding at type $\type$.

\subsection{Symmetric mechanism}
Notice that a convex combination of any two implementable allocation–utility pairs may violate the \eqref{eq:IC} constraints.
This implies that the incentive constraints \eqref{eq:IC} are not convex. 
This poses a challenge to find the optimal mechanism because the program is not convex, even though the objective function is linear. 
Hence, the conventional approach of finding extremal points of the convex program \citep{KMS2021extreme} fails to work in our setting.
We circumvent this challenge by establishing that the optimal non-coordination mechanism for this non-convex optimization problem is always symmetric in symmetric environments, as summarized in the following lemma.

\begin{lemma}\label{lmm:symmetry}
The optimal mechanism is symmetric for any $\alpha\in[0,1]$. 
\end{lemma}

By restricting our attention to symmetric mechanisms, 
the problem of designing the optimal mechanism reduces to a single-agent optimization problem.
When there is no ambiguity, we omit the subscript $i$ from the notation for this single-agent problem; we use the interim allocation rule $\alloc$ and the utility function $\util$ for a single agent to refer to the interim allocation profile and the interim utility profile, respectively. Let $\efficientalloc(\type)$ 
be the interim allocation rule maximizing matching efficiency, i.e., the efficient allocation rule. 
The optimization problem can then be reformulated as follows: 

\begin{equation}
\tag{$\hat{\cP}_{\alpha}$}
\label{eq:P_alpha single}
    \begin{aligned}
     \hat{V}_{\alpha} = \sup_{\alloc,\util} \quad& \expect[\type]{\alpha \cdot \type\cdot \alloc(\type)
     + (1-\alpha)\cdot \util(\type)} \\
     \text{s.t.} \quad& \alloc \text{ is monotone and feasible},\\
     \quad& (\alloc,\util) \text{ satisfies }\eqref{eq:IC}.
    \end{aligned}
\end{equation}
Recall that not all interim allocation rules are feasible.
The set of feasible interim allocations has been characterized in \citet{border1991implementation,che2013generalized}.
See \cref{lem:border} for a detailed discussion.

Note that, as shown in \citet{li2025mechanism}, any symmetric non-coordination mechanism can be implemented as a coarse ranking contest in arbitrary symmetric environments. Moreover, in our setting, there is a one-to-one mapping between type reports and signals. As a result, any implementable direct mechanism is essentially an all-pay contest. Therefore, for the rest of the paper, we often refer to a non-coordination mechanism as a contest.

\subsection{Optimal Contest}
Now we are ready to characterize the optimal contest. Since the optimal mechanism is symmetric. We denote the interim allocation-utility pair by  $(\optimalalloc,\optimalutil)$ for each agent, for any $\alpha\in(0,1)$.
\begin{theorem}
\label{thm:optimal characterization}
For any $\alpha\in(0,1)$,  the optimal mechanism is a contest that partitions each agent's type  space into intervals 
$\{(\underline{\type}^{(j)}, \bar{\type}^{(j)})\}_{j=1}^{\infty}$,\footnote{If the partition is finite, say, consisting of only $K$ disjoint intervals, then  $\underline{\type}^{(j)}=\bar{\type}^{(j)}$ for all $j> K$.}
such that each interval $(\underline{\type}^{(j)}, \bar{\type}^{(j)})$ belongs to 
exactly one of the following three regions:%
\footnote{The definitions of the interim allocation and utility on the cutoff points $\{\underline{\type}^{(j)}\}_{j=1}^{\infty}$ do not affect the objective value since they have zero measure.}
\begin{enumerate}[(1)]
    \item \emph{No-tension region}: $\optimalalloc(\type)=\optimalutil(\type)=\efficientalloc(\type)$ 
    and $\optimalutil'(\type) < \ability$ 
    for any type $\type\in (\underline{\type}^{(j)},\bar{\type}^{(j)})$.
    
    \item \emph{No-effort region}: $\optimalalloc(\type)=\optimalutil(\type)$ and $\optimalutil'(\type) = \ability$ for any type $\type\in (\underline{\type}^{(j)},\bar{\type}^{(j)})$, and 
        \begin{align*}
        \int_{\underline{\type}^{(j)}}^{\bar{\type}^{(j)}} \optimalalloc(\type) \dd \dist(\type) 
        = \int_{\underline{\type}^{(j)}}^{\bar{\type}^{(j)}} \efficientalloc(\type) \dd \dist(\type).
        \end{align*}

    \item \emph{Efficient region}: $\optimalalloc(\type)=\efficientalloc(\type)> \optimalutil(\type)$ and $\optimalutil'(\type) = \ability$ for any type $\type\in (\underline{\type}^{(j)},\bar{\type}^{(j)})$.
    
\end{enumerate}
\end{theorem}

The principal's problem is an optimization problem with two constraints: interim feasibility (IF) and incentive compatibility (IC).
Note that the interim feasibility constraints follow the standard Border's characterization, which is provided in detail in \cref{subapx:interim_feasible}.
Thus, a necessary condition for a mechanism to be optimal is that either one of the constraints binds or neither of them binds.
When (IF) binds, the optimal allocation rule and the efficient allocation rule coincide. 
In the no-tension region, the efficient allocation rule can be implemented with zero effort because the marginal cost of effort outweighs the marginal benefit of increased allocation consumption.
However, when the slope of the efficient allocation rule is steep enough, the marginal benefit of getting a higher allocation outweighs the marginal cost of exerting effort.
As a result, (IC) requires that agents exert a positive effort to implement the efficient allocation.
When this is the case, the principal needs to decide whether to implement the efficient allocation rule.
In the efficient region, the efficient allocation rule is implemented, and agents are required to exert certain levels of effort to maintain incentive compatibility. Both (IC) and (IF) bind in this region.
In the no-effort region, an allocation rule with a flattened slope is implemented, so that agents have no incentives to exert effort. In this region, (IC) binds, and the optimal allocation and utility coincide since no agent exerts positive effort.
\cref{tb:3categories} summarizes these possibilities.

\begin{table}[t]
\begin{tabular}{|c|c|c|c|c|}
\hline
           & IC binds     & IF  binds    &  effort   & $\optimalalloc$    \\ \hline
no-tension region& $\times$   & $\checkmark$ & $=0$     & $=\efficientalloc$     \\ \hline
no-effort region & $\checkmark$ & $\times$   & $=0$ & $\neq \efficientalloc$ \\ \hline
efficient region & $\checkmark$ & $\checkmark$ & $>0$ & $=\efficientalloc$    \\ \hline
\end{tabular}
\caption{Regions of the type space for the optimal mechanism}\label{tb:3categories}
\end{table}

The key observation from the characterization is that, when there are tensions between allocating efficiently and maintaining incentive compatibility, the principal either implements the efficient allocation, inducing agents to exert high effort, or  eliminates agents' incentives to exert effort completely. Intermediate scenarios where the principal randomizes the allocations to partially reduce effort are never optimal. 

The allocation rule can be quite complex, as both the number of intervals in each region and their ordering may exhibit complex dependencies on the shape of the efficient allocation rule and the coefficient $\alpha$.
Nonetheless,  \cref{thm:optimal characterization} enables us to obtain a sharper characterization of the optimal mechanism in large markets (see \cref{sec:largecontests}). 

We outline the main steps to derive \cref{thm:optimal characterization} and defer the formal proof to \cref{apx:A3_IC}.
First, \cref{lmm:symmetry} reduces the original $n$-agent problem to the single-agent program \eqref{eq:P_alpha single}.
The constraint set combines interim feasibility---which in symmetric monotone environments is equivalent to the cumulative majorization constraint \eqref{eq:feasibility symmetric} (see \cref{lem:border} in the Appendix)---and the novel \eqref{eq:IC} restrictions from \cref{lem:monotone alloc-util implementation}.
Because \eqref{eq:IC} is non-convex, standard convex-program/extreme-point methods do not apply.

It is tempting to pursue a Myersonian reduction that ``eliminates'' one instrument.
In classical quasilinear screening, IC yields an envelope representation so transfers can be substituted out, leaving a pointwise (virtual-surplus) problem in the allocation.
Here, although IC does pin down the best implementable utility for a given interim allocation $\alloc$ (cf.\ \cref{cor:payoff equiv} and \eqref{eq:utils}), the resulting map $\alloc\mapsto \util$ is nonlocal (it involves an $\inf_{\type'\le \type}$ operator), so eliminating $\util$ does not lead to a tractable pointwise objective in $\alloc$.
Second, the alternative reinterpretation suggested by the quasilinear analogy also turns out to be unhelpful in our setting.
One may try to view the signal as the ``allocation'' object and the interim allocation $\alloc$ as the ``transfer'' object, mirroring screening with quasilinear transfers.
But then interim feasibility \eqref{eq:feasibility symmetric} becomes a feasibility constraint on the transfer schedule (since $\alloc$ must be implementable as an interim allocation from an ex-post $k$-unit allocation rule).
Standard Myerson tools rely crucially on transfers being free of such feasibility restrictions; with a continuum of Border-type majorization inequalities imposed on the transfer-like object, there is no direct virtual-value reduction or ironing argument to recover an optimal solution.\footnote{\citet{perez2024score} provides a characterization of the optimal mechanism in a single-agent setting, which avoids such a challenge.} 

To overcome these challenges, we apply optimal control techniques.
Write feasibility in terms of the cumulative allocation state $\intalloc(\type):=\int_{\type}^{\bar{\type}}\alloc(t)\dd\dist(t)$, so that \eqref{eq:feasibility symmetric} becomes the state constraint
$\intalloc(\type)\le \int_{\type}^{\bar{\type}}\efficientalloc(t)\dd\dist(t)$ for all $\type$, with $\intalloc'(\type)=-\alloc(\type)\density(\type)$.
Treat the utility slope $u(\type):=\util'(\type)\in[0,\ability]$ as the control, and rewrite the IC level constraint $\util(\type)\le \alloc(\type)$ as $\util(\type)\density(\type)+\intalloc'(\type)\le 0$.
Euler–Lagrange conditions plus complementary slackness yield a ``bang–bang'' implication: whenever feasibility is slack, the level constraint binds ($\util=\alloc$) and the control must hit an endpoint ($u\in\{0,\ability\}$); see \cref{lmm:constraint}.
This then pins down where feasibility must bind: if the mechanism induces effort ($\alloc>\util$), feasibility binds and $\alloc=\efficientalloc$ almost everywhere (\cref{cor:case3}); likewise, if $0<u<\ability$, feasibility binds and again $\alloc=\efficientalloc$ almost everywhere (\cref{cor:case1a}).
These implications restrict the optimum to the three patterns in \cref{thm:optimal characterization} (no-tension, no-effort, efficient), and a perturbation argument rules out flat segments with $u=0$ (\cref{lem:no case 1b}).
In particular, mechanisms that randomize allocations to \emph{partially} reduce effort (without eliminating it) are never optimal.

\section{Large Contests}
\label{sec:largecontests}
In this section, we show that the optimal mechanism exhibits particularly simple structures when the population of participants is large. 
To simplify the exposition, we make an additional assumption throughout this section.\footnote{The main economic insights extend without this assumption.}

\begin{assumption}[continuity]
\label{asp:continuity}
There exist $\underline{\beta}_1,\bar{\beta}_1,\beta_2 \in (0,\infty)$ such that $\density(\type) \in [\underline{\beta}_1,\bar{\beta}_1]$ and $\density'(\type) \geq -\beta_2$ 
 for any type $\type\in[\underline{\type},\bar{\type}]$.
\end{assumption}

\subsection{Scarce Resources}
\label{sec:scarce}
In some applications, such as the awarding of prestigious fellowships to university students, 
the competition is fierce 
and the ratio of the number of competing agents to the number of items available is large. 
In this subsection, we study the model of \cref{sec:model} in the special case where $k=1$ 
and the number of agents $n$ is very large.\footnote{When $k$ is a small constant greater than $1$, the analysis is significantly more involved. We omit this case here since the economic insights it yields are similar.} 
For sufficiently large $n$, the efficient allocation rule becomes convex, which simplifies the characterization of the optimal contest. The proofs of the results in this subsection are provided in \cref{apx:A5_scarce}.

\begin{lemma}\label{lem:large_convex}
Let $k=1$. Under \Cref{asp:continuity},
there exists a positive integer $N$
such that for any $n\geq N$, 
the efficient allocation rule $\efficientalloc(\type)$ is convex in $\type$. 
\end{lemma}

\paragraph{Convex efficient allocation.}
First consider the case when the efficient allocation rule is convex. 
The optimal contest is then as follows.

\begin{proposition}
\label{prop:convex optimal}
Suppose $\efficientalloc(\type)$ is convex in $\type$. 
For any $\alpha\in(0,1)$, there exist cutoff types $\underline{\type}\leq\type^{(1)}\leq\type^{(2)}\leq\bar{\type}$
such that in the type space of each agent in the optimal contest $\optimalalloc$, the interval $(\underline{\type},\type^{(1)})$ is the no-tension region, $(\type^{(1)},\type^{(2)})$ is the no-effort region, and $(\type^{(2)},\bar{\type})$ is the efficient region.
 \end{proposition}

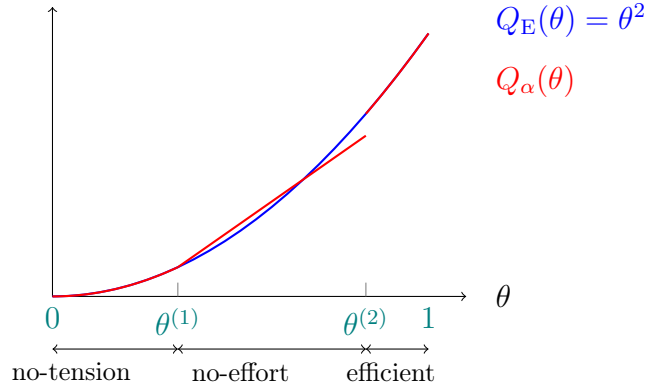
\begin{figure}[t]
\centering
\begin{tikzpicture}[xscale=5,yscale=3.5]

\draw [<->] (0,1.1) -- (0,0) -- (1.1,0);
\node [right] at (1.15, 0 ) {$\type$};
\node [below, teal, thick]  at (0, 0 ) {$0$};
\node [below, teal, thick]  at (1, 0 ) {$1$};

\draw [blue, domain=0:1, thick] plot (\x, {\x^2});
\node [below, teal, thick] at (1/3, 0 ) {$\type^{(1)}$};
\node [below,teal,thick] at (5/6, 0 ) {$\type^{(2)}$};
\node [right, blue] at (1.15, 1.05 ) {$\efficientalloc(\type)=\type^2$};

\draw [red, domain=0:1/3, thick] plot (\x, {\x^2-0.00});
\draw [red, domain=1/3:5/6, thick] plot (\x, {\x-1/3+(1/3)^2-0.00});
\draw [red, domain=5/6:1, thick] plot (\x, {\x^2-0.00});
\node [right, red] at (1.15, 0.82 ) {$\optimalalloc(\type)$};

\draw [gray] (4/12,0) -- (4/12, 0.0525);
\draw [gray] (10/12,0) -- (10/12, 0.0525);

\draw [<->] (0,-0.2) -- (4/12,-0.2);
\draw [<->] (4/12,-0.2) -- (10/12,-0.2);
\draw [<->] (10/12,-0.2) -- (1,-0.2);

\node [below] at (0.05,-0.21) {\footnotesize{\text{no-tension}}};
\node [below] at (0.5,-0.21) {\footnotesize{\text{no-effort}}};
\node [below] at (0.9,-0.21) {\footnotesize{\text{efficient}}};
\end{tikzpicture}
\caption{Optimal interim allocation rule under convex $\efficientalloc(\type)$ }
\label{fig:convex_opt}
\rule{0in}{1.2em}$^\dag$\scriptsize Suppose $n=2$, $k=1$, $\dist(\type)=\type^2$, $\type\in [0,1]$, and $\ability=1$. In this example, the interim efficient allocation rule is $\efficientalloc(\type)=\dist^{n-1}(\type)=\type^2$, i.e., the highest type gets the item, and $\optimalalloc(\type)$ is the optimal interim allocation rule.\\
\end{figure}

\Cref{fig:convex_opt} illustrates the optimal interim allocation rule arising from a convex efficient allocation rule in an example with two agents. \Cref{fig:implement_convex_efficient} illustrates the corresponding ex post rule that maps type profiles to allocations.
Intuitively, when the efficient allocation is convex, its derivative crosses $\ability$ from below only once. 
Therefore, for low types, there is no tension: the derivative of the efficient allocation is small enough so that, in the optimal contest, the item can be allocated efficiently without any effort on the part of the agents. 
For high types, since the change in the efficient allocation is large, the incentive constraints bind and the interim utility must be linear. 
Moreover, in order for the interim allocation to be interim feasible, 
in the region where the utility is linear, the no-effort region must occur before the efficient region, not the other way around.

 \begin{figure}[t]
\centering
\begin{tikzpicture}[xscale=5,yscale=5]

\draw [<->] (0,1.1) -- (0,0) -- (1.1,0);
\draw  (0,1) -- (1,1) -- (1,0);
\node [right] at (0, 1.1 ) {$\type_2$};
\node [right] at (1.15, 0 ) {$\type_1$};
\draw [gray, domain=0:1, dashed] plot (\x, {\x});

\draw [gray] (4/12,0) -- (4/12, 0.0525);
\draw [gray] (10/12,0) -- (10/12, 0.0525);

\node [below, brown, thick]  at (0, 0 ) {$0$};
\node [below, brown, thick] at (1/3, 0 ) {$\type^{(1)}$};
\node [below,brown,thick] at (5/6, 0 ) {$\type^{(2)}$};
\node [below, brown, thick]  at (1, 0 ) {$1$};

\draw [gray] (0, 4/12) -- (0.0525, 4/12 );
\draw [gray] (0, 10/12) -- (0.0525, 10/12);

\node [left, teal, thick] at (0,1/3 ) {$\type^{(1)}$};
\node [left,teal,thick] at (0, 5/6 ) {$\type^{(2)}$};

\draw[blue, thick] (1/3,5/6) -- (1/3,1/3) -- (5/6,1/3) -- (5/6,5/6) -- (1/3,5/6);
\node [blue] at (3/5,3/5) {\text{randomized}};

\draw[teal, domain=0:1/3-0.01, thick] plot (\x, {\x+0.01});
\draw[teal, domain=(5/6):(1-0.01), thick] plot (\x, {\x+0.01});
\draw[teal, thick] (5/6,5/6+0.01) --(1/3-0.01,5/6+0.01) -- (1/3-0.01,1/3) ;
\node [teal] at (2/6,0.9) {\text{agent $2$ wins}};

\draw[brown, domain=0.01:1/3, thick] plot (\x, {\x-0.01});
\draw[brown, domain=(5/6+0.01):1, thick] plot (\x, {\x-0.01});
\draw[brown, thick] (5/6+0.01,5/6) --(5/6+0.01,1/3-0.01) -- (1/3,1/3-0.01) ;
\node [brown] at (3/5,1/6) {\text{agent $1$ wins}};

\end{tikzpicture}
\caption{Implementation of the optimal allocation rule}
\label{fig:implement_convex_efficient}
\rule{0in}{1.2em}$^\dag$\scriptsize Suppose $n=2$, $k=1$, $\dist(\type)=\type^2$, $\type\in [0,1]$. When both agents produce signals in $(\type^{(1)},\type^{(2)})$, the item is allocated randomly, but the agent with the higher signal has a higher probability (which is strictly less than $1$) of getting the item. When at least one agent produces a signal outside $(\type^{(1)},\type^{(2)})$, the item is allocated to the agent with a higher signal.\\
\end{figure}
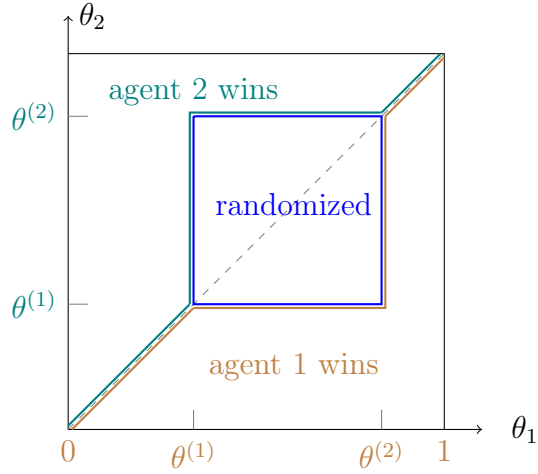

It is interesting to note that in the optimal contest when the efficient allocation rule is convex,
there is distortion for middle types but not for high or low types. This stands in sharp contrast to the classical auction design setting, where distortions typically occur for low types.

\paragraph{Convergence results.}
Using \cref{lem:large_convex} and \cref{prop:convex optimal}, we can immediately characterize the optimal contest for the allocation of scarce resources across a large number of agents. 
Moreover, we show that as the number of agents increases, 
the no-tension region converges to the full type space. 
Since the contest format in the no-tension region is WTA,
this implies that in the limit, 
the format of the optimal contest is essentially WTA.

\begin{theorem}[convergence of contest format]
\label{thm:scarce_resource}
Let $k=1$. Under \Cref{asp:continuity}, for any $\alpha\in (0,1)$,
there exists $N$ such that for any $n\geq N$, 
the optimal contest
takes the form described in \cref{prop:convex optimal}.
Moreover, as $n$ goes to infinity, the no-tension region converges to the entire type space. 
\end{theorem}

Given this convergence result, it may appear tempting to host a WTA contest, since it approximates the optimal contest for a large finite number of agents. However, as shown in \cref{thm:scarce_nonconvergence} below, WTA contest is not a good approximation of the optimal contest.

Before we formally state our result, we introduce one more piece of notation.
For any interim allocation rule~$\alloc$ that is implementable by a contest, by \cref{cor:payoff equiv}, there exists a unique interim utility  $\util$ with $\util(\underline\type)=\alloc(\underline\type)$ such that $(\alloc,\util)$ is implementable by a contest and yields a weakly higher payoff for the principal than any other mechanism.
Denote this payoff by $V_{\alpha}(\alloc)$, i.e.,
\begin{align*}
    V_{\alpha}(\alloc)=\sup \{\text{Obj}_{\alpha}(\alloc,\util): \util(\underline\type)=\alloc(\underline\type) \text{ and }(\alloc,\util) \text{ is implementable by a contest}\}.
\end{align*}

\begin{theorem}[non-convergence in payoffs]
\label{thm:scarce_nonconvergence}
Let $k=1$. Under \Cref{asp:continuity}, for any $\alpha\in (0,1)$ and any sufficiently small $ \epsilon > 0$, 
there exists $N_{\dist,\epsilon}\in \posnaturals$ and $\delta > 1$ such that for any finite $n>N_{\dist,\epsilon}$,
the ratio between the principal's payoff in the optimal contest and her payoff in the WTA contest is at least $\delta$;
that is, $\frac{V_{\alpha}(\optimalallocn{n})}{V_{\alpha}(\efficientallocn{n})}\geq \delta> 1$.
\end{theorem}

This result states that the principal's payoff under the WTA contest does not converge to her payoff under the optimal contest as the number of agents increases. 
By randomizing the allocation for a small range of high types, whose measure is of order $\frac{1}{n}$, its contribution to the principal's objective value is large.
This is because the probability that there exists an agent with a type falling into this range is roughly of order $\frac{1}{e}$.
Such an agent would have exerted high effort if it were under the efficient allocation rule, which would have created a large decrease in the agents' utilities and hence the principal's objective value. 
Thus, as long as the planner's payoff has a non-zero weight on agents' utilities, 
randomly allocating the scarce resource to any agent with sufficiently high signals can substantially increase the agents' expected utilities while keeping the loss in matching efficiency small. 
Conceptually, this contrasts with \citet{bulow1996auctions}, who find that the revenue from the efficient allocation rule with one additional buyer exceeds that under the optimal mechanism.

Note that our framework enables us to completely characterize the optimal contest for a large but finite number of agents. For comparison, \citet{OS2016large,olszewski2020performance} approximate equilibria in contests using a continuum model. Our finding that the optimal contest format converges to the WTA format is consistent with the results of \citet{OS2016large}, 
but the non-convergence of the optimal payoff (\cref{thm:scarce_nonconvergence}) stands in contrast to their work. 
The non-convergence in payoff result highlights the importance of randomizing the allocations for top types in practical applications to reduce the cost of wasteful signals.

\subsection{Large-Scale Economy}
\label{sec:large scale}
In applications such as college admissions and affordable housing programs, the resources to be allocated are not necessarily scarce.  
To model such situations, consider a setting with $n$ agents and $0<k< n$ items, and replicate both the agents and the items $z\in \N_+$ times. 
The parameter $z$ captures the scale of the economy.
As the scale $z$ goes to infinity, the efficient allocation rule in this setting converges to the cutoff rule, under which the items are allocated to the top $\frac{k}{n}$ of the types. 
Efficient allocation hence creates strong incentives for types close to the cutoff to exert wasteful effort.
\cref{thm:large_scale} shows that in the optimal contest, to eliminate these incentives, the principal randomizes the allocation for types close to the cutoff.
Let $\type_c$ be the cutoff type, defined by $\Pr[\type\geq \type_c] = \frac{k}{n}$.

\begin{theorem}\label{thm:large_scale}
Under \Cref{asp:continuity}, for any $\alpha\in (0,1)$ 
and any fixed integers $n>k>0$, there exists $Z$ such that for any integer $z\geq Z$, 
in a setting with $z\cdot n$ agents and $z\cdot k$ items, there exist cutoff types $\underline{\type}\leq\type^{(1)}<\type_c< \type^{(2)}\leq \type^{(3)}\leq \bar{\type}$ such that
in the type space of each agent in the optimal contest, the intervals $(\underline{\type},\type^{(1)})$ and $(\type^{(3)},\bar{\type})$ comprise the no-tension region, $(\type^{(1)},\type^{(2)})$ is the no-effort region, and $(\type^{(2)},\type^{(3)})$ is the efficient region.
\end{theorem}
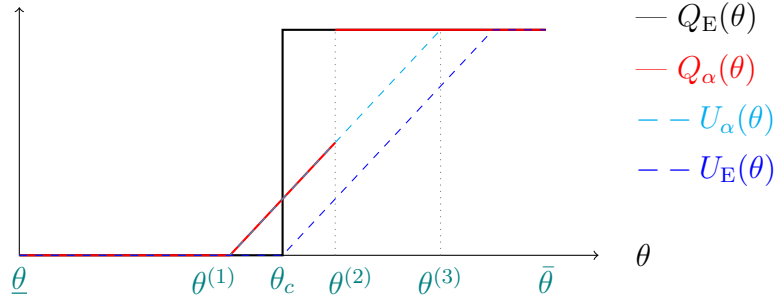
\begin{figure}[t]
\centering
\begin{tikzpicture}[xscale=7,yscale=3]

\draw [<->] (0,1.1) -- (0,0) -- (1.1,0);

\draw [thick] (0,0) -- (0.5,0) -- (0.5,1) -- (1,1);
\draw [red, thick] (0,0) -- (0.4,0) -- (0.6,0.5);
\draw [red, thick] (0.6,1) -- (1,1);
\draw [cyan, dashed] (0,0) -- (0.4,0) --(0.6,0.5) -- (0.8,1);
\draw [blue, dashed] (0,0) -- (0.5,0) -- (0.9,1) -- (1,1);

\node [right] at (1.15, 0 ) {$\type$};
\node [below, teal, thick]  at (0, 0 ) {$\underline{\type}$};
\node [below, teal, thick]  at (1, 0 ) {$\bar{\type}$};
\node [below, teal, thick]  at (0.5, 0 ) {$\type_c$};
\node [below, teal, thick]  at (0.37, 0 ) {$\type^{(1)}$};
\node [below, teal, thick]  at (0.63, 0 ) {$\type^{(2)}$};
\node [below, teal, thick]  at (0.8, 0 ) {$\type^{(3)}$};

\node [right, black] at (1.15, 1.05 ) {---\,\,$\efficientalloc(\type)$};
\node [right, red] at (1.15, 0.82 ) {---\,\,$\optimalalloc(\type)$};
\node [right, cyan] at (1.15, 0.6 ) {$--\optimalutil(\type)$};

\node [right, blue] at (1.15, 0.38 ) {$--\efficientutil(\type)$};

\draw [gray,dotted] (0.6,0) -- (0.6, 1);
\draw [gray,dotted] (0.8,0) -- (0.8, 1);

\end{tikzpicture}
\caption{Optimal allocation and utility for large-scale economy in the limit}
\rule{0in}{1.2em}$^\dag$\scriptsize
Types in $(\type^{(1)},\type^{(2)})$ are called middle types; in the optimal contest, their utilities are higher than under efficient allocation, because they do not exert effort. Types above $\type^{(2)}$ are called high types; in the optimal contest their utilities are weakly greater, or in some cases strictly greater, because they each exert less effort.
\label{fig:large_scale}
\end{figure}

Intuitively, in the limit, the efficient allocation rule converges to a step function, 
with only types above $\type_c$ receiving the items. 
The interim utility under efficient allocation is thus represented by the blue curve in \Cref{fig:large_scale}. 
In order to increase the weighted average between the matching efficiency and the sum of the expected utilities, 
the principal can randomize the allocation for types around the cutoff $\type_c$,
i.e., within $(\type^{(1)},\type^{(2)})$. 
This leads to an efficiency loss of at most $\type^{(2)}-\type^{(1)}$ when an item is allocated inefficiently, but increases the utilities for all types within $(\type^{(1)},\type^{(3)})$. 
When 
$\type^{(1)}$ is sufficiently close to $\type^{(2)}$, the increase in expected utility is significantly larger than the efficiency loss. 
Therefore, the principal can increase her payoff by randomizing on $(\type^{(1)},\type^{(2)})$. 
Finally, for types that are sufficiently low or sufficiently high, 
it is easy to verify that both the matching efficiency and the sum of the agents' utilities are maximized under efficient allocation. 
In \cref{apx:A6_large_scale}, we show that this intuition applies when the scale of the economy is finite but sufficiently large.

Our result is reminiscent of Director's law,
which states that public programs tend to be designed primarily to target the middle classes. 
Specifically, although the principal cares about the utilities of all of the agents, the optimal contest gives preferential treatment to the middle types, in the sense that they obtain higher utilities than they would in a fully competitive setting where items are allocated efficiently. 
This is optimal for the principal because it induces the middle types to exert no effort, which weakly (strictly) decreases the effort level for all (some) of the higher types.
This reasoning is also in line with the empirical results of \citet{KLOST2022pareto}, who show (using data from Turkey) that randomizing the allocation of college seats to students, especially those with low scores, reduces overall student stress.

\section{Conclusions}
\label{sec:conclusion}
Strategic misrepresentation often occurs in institutions with weak oversight. We incorporate strategic misrepresentation behaviors in resource allocation problems and characterize the optimal mechanism that accounts for such behaviors. 
We apply our theoretical insights to resource allocation in public programs distinguished by resource scarcity. When the population of the participants is large, we show that randomizing allocation to middle types strikes the best balance between matching efficiency and a utilitarian objective. Our work opens the door to several interesting research directions. First, in scenarios where misrepresented information is not verifiable, the question of whether randomization is still optimal when agents are equipped with different misrepresentation technologies remains open. Second, although
empirical studies show a strong connection between institutional oversight and incentives for misrepresentation \citep{bandiera2009active}, little is known on the theoretical implications of the interaction between the two forces.

\bibliographystyle{apalike}
\bibliography{references}
\newpage

\appendix
\section{Proof of \cref{thm:optimal characterization}}
\label{apx:A3_IC}

\subsection{Interim Feasibility}
\label{subapx:interim_feasible}
When using the interim approach, it is important to restrict attention to feasible interim allocation rules.  
\begin{definition}[feasibility]
Given any interim allocation rule $\allocs=(\alloc_i)_{i=1}^n$ where $\alloc_i:\Theta_i\rightarrow [0,1]$, 
$\allocs$ is \emph{feasible} if there exists an ex-post allocation rule $\contestallocs$
such that 
$\alloc_i(\type_i) = \expect[\type_{-i}]{\contestalloc_i(\type_i, \type_{-i})}$ for any agent~$i$ and type $\type_i$.
We say an interim allocation--utility pair $(\allocs,\utils)$ is \emph{feasible} if $\allocs$ is feasible. 
\end{definition}
 
\citet{border1991implementation} and \citet{che2013generalized} have provided closed-form characterization for feasible interim allocations. 
\begin{lemma}[\citealp{che2013generalized}]
\label{lem:border}
Given a set $\boldsymbol{A}=\prod_{i=1}^nA_i\subset \typespace$, let $w(\types,\boldsymbol{A}) = |\lbr{i : \type_i\in A_i}|$ be the number of agents whose type $\type_i$ is in $A_i$.\footnote{Here, $|\cdot|$ denotes the cardinality of a set.}
The interim allocation rule $\allocs$ is interim feasible if and only if 
\vspace{-8pt}
\begin{align*}
\tag{IF}\label{eq:feasibility}
\sum_i \int_{A_i}  \alloc_i(\type_i) \dd \dist_i(\type_i)
\leq \int_{A} \min\lbr{k,w(\types,\boldsymbol{A})} \dd \dists(\type)
\qquad \forall \boldsymbol{A}=\prod_{i=1}^nA_i\subset \typespace.
\end{align*}
Moreover, for monotone allocations in symmetric environments, 
\eqref{eq:feasibility} is equivalent to the following:\footnote{In a symmetric environment, by a slight abuse of notation, we use $\dist=\dist_i$ for all $i$ to denote each agent's type distribution.}
\begin{align}\label{eq:feasibility symmetric}
\tag{$\widehat{\mathrm{IF}}$}
\int_{\type}^{\bar{\type}} \alloc(z)\dd \dist(z) \leq \int_{\type}^{\bar{\type}} \efficientalloc(z) \dd\dist(z), \quad\forall \type\in[\underline{\type},\bar{\type}],
\end{align}
where $\efficientalloc(\type)=\sum_{j=0}^{k-1}\binom{n-1}{j} 
\cdot (1-\dist(\type))^j\cdot \dist^{n-1-j}(\type)$ is the interim allocation rule for allocating $k$ items efficiently. 
\end{lemma}

\subsection{Proof of \cref{lem:monotone alloc-util implementation}}
We will prove each direction of the if-and-only-if condition separately. 

\paragraph{Only if:}
If $(\allocs,\utils)$ is implementable by a non-coordination mechanism, 
there exist a signal recommendation policy $\reportsignals$ and
an allocation rule $\signalallocs$ that induce $(\allocs,\utils)$.
The allocation rule~$\allocs$ satisfies interim feasibility, because it is induced by the ex-post allocation rule $\contestalloc_i(\types)=\signalalloc_i(\reportsignal_i(\type_i),\reportsignals_{-i}(\type_{-i}))$ for all $i$ and $\types$.
Notice that for any signal recommendation policy $\reportsignals$ and
allocation rule $\signalallocs$ implementing $(\allocs,\utils)$, 
it is without loss to assume that any realization of the recommendation $\reportsignal_i(\type_i)$ is weakly higher than $\type_i$. 
This is because weakly increasing a signal recommendation below $\type_i$ does not induce type $\type_i$ to exert additional effort, 
but weakly decreases other types' incentives for deviation.

For any agent $i$ and any pair of types $\type_i < \type'_i$, 
let $\signal'_i$ be the largest signal realization given $\reportsignal_i(\type'_i)$.
Thus we have $\signal'_i\geq\type'_i$.
Note that agent $i$ with type $\type'_i$ obtains utility $\util_i(\type'_i)$ from choosing signal $\signal'_i$,
as he must be indifferent for all his signal realizations. 
Therefore, agent $i$'s utility from reporting signal $\signal'_i$ when his type is $\type_i$ is 
\begin{align*}
\expect[\type_{-i}]{\signalalloc_i(\signal'_i,\reportsignal_{-i}(\type_{-i}))} - \ability\cdot (\signal'_i-\type_i)^+
&= \expect[\type_{-i}]{\signalalloc_i(\signal'_i,\reportsignal_{-i}(\type_{-i}))} - \ability\cdot (\signal'_i-\type'_i)^+
- \ability\cdot (\type'_i-\type_i)\\
&= \util_i(\type'_i) - \ability\cdot (\type'_i-\type_i).
\end{align*}
Since his utility from deviating in his choice of signal is weakly lower, 
we have
\begin{align*}
\util_i(\type_i)\geq \util_i(\type'_i) - \ability\cdot (\type'_i-\type_i).
\end{align*}
By rearranging the terms and taking the limit as $\type'_i\to \type_i$, we obtain $\util'_i(\type) \leq \ability$.
Similarly, let $\signal_i$ be the largest signal realization given $\reportsignal_i(\type_i)$.
We have 
\begin{align*}
\util_i(\type'_i) &\geq \expect[\type_{-i}]{\signalalloc_i(\signal_i,\reportsignal_{-i}(\type_{-i}))} - \ability\cdot (\signal_i-\type'_i)^+\\
&\geq \expect[\type_{-i}]{\signalalloc_i(\signal_i,\reportsignal_{-i}(\type_{-i}))} - \ability\cdot (\signal_i-\type_i)^+
= \util_i(\type_i).
\end{align*}
Again by rearranging the terms and taking the limit as $\type'_i\to \type_i$, we have $\util'_i(\type) \geq 0$. Hence $\util'_i(\type_i) \in[0,\ability]$ for any type $\type_i$.

Finally, as the effort is non-negative, the interim allocation must be weakly larger than the interim utility. 
When the inequality is strict, the agent must choose signal realizations strictly higher than his type with positive probability. 
In this case, we have $\signal_i > \type_i$. 
For any type $\type'_i \in (\type_i, \signal_i)$, we have 
\begin{align*}
\util_i(\type'_i) &\geq \expect[\type_{-i}]{\signalalloc_i(\signal_i,\reportsignal_{-i}(\type_{-i}))} - \ability\cdot (\signal_i-\type'_i)^+\\
&= \expect[\type_{-i}]{\signalalloc_i(\signal_i,\reportsignal_{-i}(\type_{-i}))} - \ability\cdot (\signal_i-\type_i)^+
+ \ability\cdot (\type'_i-\type_i)\\
&= \util_i(\type_i) + \ability\cdot (\type'_i-\type_i).
\end{align*}
By rearranging the terms and taking the limit as $\type'_i\to \type_i$, we obtain $\util'_i(\type_i) \geq \ability$. 
Since we also know that $\util'_i(\type_i) \leq \ability$, 
both inequalities must be equalities and hence $\util'_i(\type_i) = \ability$.

\paragraph{If:}
Since $\allocs$ is interim feasible, 
there exists an ex-post allocation rule $\contestallocs$ that implements $\allocs$. 
Consider the signal recommendation policy $\reportsignals$
where $\reportsignal_i(\type_i) = \type_i + \frac{1}{\ability}(\alloc_i(\type_i) - \util_i(\type_i))$ 
for any agent~$i$ with type~$\type_i$. 
It is easy to verify that $\reportsignal_i(\type_i)$ is monotone in $\type_i$, since $\alloc'_i(\type_i)\geq 0$ and $\util'_i(\type_i)\leq \ability$. 
Let $\type_i(\signal_i)$ be the inverse of function $\reportsignal_i$.%
\footnote{Note that $\reportsignal_i(\type_i)$ is only weakly monotone. 
When there are multiple types $\type_i$ with the same signal recommendation~$\signal_i$, 
we map $\signal_i$ randomly to those types according to the type distribution $\dist_i$.}
Consider the allocation rule $\signalallocs$ where 
$\signalalloc_i(\signals) = \contestalloc_i(\types(\signals))$ for all agents $i$. 
We show that $\reportsignals$ and $\signalallocs$ implement $(\allocs,\utils)$. 

First, by our construction, when all agents follow the recommendations, the interim allocation and the interim utility coincide with $\allocs$ and $\utils$, respectively. 
Thus, it is sufficient to show that the agents have weak incentives to follow the recommendations. 
In particular, if agent $i$ with type~$\type_i$ deviates to reporting type $\type'_i>\type_i$, 
his utility from deviation is  
\begin{align*}
\alloc_i(\type'_{i}) - \ability\cdot (\reportsignal_i(\type'_i)-\type_i)^+
= \util_i(\type'_i)
- \ability\cdot(\type'_i-\type_i)
\leq \util_i(\type_i),
\end{align*}
where the last inequality holds because the derivative of $\util$ is always at most $\ability$.
We now analyze the incentives for downward deviation in three cases. 
If the deviation type $\type'_i < \type_i$ satisfies $\alloc_i(\type'_i) = \util_i(\type'_i)$, 
the utility from deviation is 
\begin{align*}
\alloc_i(\type'_{i}) - \ability\cdot (\reportsignal_i(\type'_i)-\type_i)^+
= \util_i(\type'_i)
\leq \util_i(\type_i).
\end{align*}
If the deviation type $\type'_i < \type_i$ satisfies $\alloc_i(\type'_i) > \util_i(\type'_i)$, 
let $\type\primed_i > \type'_i$ be the smallest type such that $\alloc_i(\type\primed_i) = \util_i(\type\primed_i)$. 
If $\type_i\geq \type\primed_i$, 
the utility from deviation is 
\begin{align*}
\alloc_i(\type'_{i}) - \ability\cdot (\reportsignal_i(\type'_i)-\type_i)^+
\leq \alloc_i(\type\primed_i)
= \util_i(\type\primed_i)
\leq \util_i(\type_i).
\end{align*}
If $\type_i < \type\primed_i$, 
the derivative of $\util$ for any type between $\type'_i$ and $\type_i$ must be constant and equal to~$\ability$.
Hence, the utility from deviation is \begin{align*}
\alloc_i(\type'_{i}) - \ability\cdot (\reportsignal_i(\type'_i)-\type_i)^+
\leq \util_i(\type'_i)
+ \ability\cdot(\type_i-\type'_i)
= \util_i(\type_i).
\end{align*}
Combining these inequalities, we conclude that none of the agents has any incentive to deviate from the recommendations.

\subsection{Optimality of symmetric mechanisms.}
We first establish \cref{cor:payoff equiv}, a payoff equivalence result within the class of non-coordination mechanisms.

\begin{proposition}\label{cor:payoff equiv}
Fix any monotone and interim feasible allocation rule $\allocs$, 
and any $\{\underline{u}_i\}_{i=1}^n$ such that $\underline{u}_i\leq \alloc_i(\underline{\type}_i)$ for all~$i$.
There exists a unique interim utility profile $\utils$ with $\util_i(\underline{\type}_i) = \underline{u}_i$ for all~$i$ 
such that $(\allocs,\utils)$ is implementable by a non-coordination mechanism.

Moreover, for any interim allocation--utility pair $(\allocs,\utils\primed)$ that is implementable by a non-coordination mechanism, 
we have the following: 
\begin{itemize}
    \item If $\util_i(\underline{\type}_i) > \util\primed_i(\underline{\type}_i)$ for any agent $i$, 
    then $\util_i(\type_i) \geq \util\primed_i(\type_i)$ for every agent $i$ and every type $\type_i$.
    
    \item If $\util_i(\underline{\type}_i) = \util\primed_i(\underline{\type}_i)$ for any agent $i$, 
    then $\util_i(\type_i) = \util\primed_i(\type_i)$ for every agent $i$ and every type $\type_i$.
\end{itemize}
\end{proposition}

In the classical mechanism design setting, payoff equivalence means that once the allocation is determined, 
the curvature of the utility function is fixed, 
and the utility function can only be shifted by a constant determined by the utility of the lowest type.
In our setting, for a fixed allocation rule, shifting the utility of the lowest type does not shift the utilities for all types by the same constant. 
We illustrate this in \Cref{fig:unique u}.
For any agent $i$, if the utility of the lowest type is lower than the interim allocation of the lowest type, or if the derivative of the interim allocation is larger than the parameter $\ability$,  then \eqref{eq:IC} implies that 
the interim utility $\util_i$ must be a straight line with derivative $\ability$
until $\util_i$ intersects $\alloc_i$ (in the example in \Cref{fig:unique u}, the intersection occurs at type $\type_i^{(1)}$). 
Then $\util_i$ coincides with $\alloc_i$ until the derivative of $\alloc_i$ exceeds $\ability$. 
In a setting with discrete types, one could apply this reasoning to recursively pin down the interim utility for all types.
Unfortunately, the recursive argument fails to work when the type space is continuous and we provide a formal proof to circumvent this technicality.

\begin{figure}[t]
		\centering
		\begin{tikzpicture}[xscale=5,yscale=4,
  pics/legend entry/.style={code={%
        \draw[pic actions] 
        (-0.25,0.25) -- (0.25,0.25);}}]
\draw [<->] (0,1.5) -- (0,0) -- (1.6,0);
\draw [blue, domain=0:1.5, thick] plot (\x,{sqrt(\x+0.04)});

\draw [red, domain=0:0.6, thick] plot (\x,{\x+0.2-0.003});
\draw [red, domain=0.6:1.5, thick] plot (\x,{sqrt(\x+0.04)-0.003});
\draw [teal, domain=0:0.96, thick] plot (\x,{\x+0.04-0.006});
\draw [teal, domain=0.96:1.5, thick] plot (\x,{sqrt(\x+0.04)-0.006});

\draw [dashed, gray] (0.6,0.8) -- (0.6,0) ;
\draw [dashed, gray] (0.96,1) -- (0.96,0) ;
\node [below,red] at (0.6, 0 ) {$\type_i^{(1)}$};
\node [below,teal] at (0.96, 0 ) {$\type_i^{\dagger (1)}$};

\node [below] at (0, 0 ) {$\underline\type_i$};
\node [below] at (1.55, 0 ) {$\bar\type_i$};
\node [below] at (1.65, 0 ) {$\type_i$};
\node [left, teal] at (-0.02, 0 ) {\footnotesize{$\util_i\primed(\underline\type_i)$}};
\node [left, red] at (-0.02, 0.2 ) {\footnotesize{$\alloc_i(\underline\type_i)=\util_i(\underline\type_i)$}};

\matrix [draw, above right] at (1.6,1) {
 \pic[blue]{legend entry}; &  \node[blue,font=\tiny] {$\alloc_i$}; \\
 \pic[red]{legend entry}; &  \node[red,font=\tiny] {$\util_i$}; \\
 \pic[teal]{legend entry}; &  \node[teal,font=\tiny] {$\util_i\primed$}; \\
};

\end{tikzpicture}
\caption{Illustration of \cref{cor:payoff equiv}}
\rule{0in}{1.2em}$^\dag$\scriptsize Both $\util_i$ and $\util_i\primed$ implement the allocation rule $\alloc_i$, but $\util_i$ gives agent $i$ a higher utility and hence is the better implementation from the principal's point of view. Moreover, $\util_i$ and $\util_i\primed$ do not differ by a constant as in the standard payoff equivalence result (where the constant would equal $\util_i(\underline\type_i)-\util_i\primed(\underline\type_i)$). However, by the construction provided in the proof of \cref{cor:payoff equiv}, $\util_i$ is uniquely identified by the allocation rule and $\util_i(\underline\type_i)$.\\
    \label{fig:unique u}
	\end{figure}

It is immediate from \cref{cor:payoff equiv} that the lowest type exerts zero effort in the optimal non-coordination mechanism, because if we set the utility of the lowest type equal to the interim allocation, then
the utilities of all higher types are weakly increased. 

\begin{proof}[Proof of \cref{cor:payoff equiv}]
For any agent $i$, 
given any monotone and interim feasible allocation $\allocs$, 
and $\underline{u}_i \leq \alloc_i(\underline{\type}_i)$, 
let 
\begin{align}
\label{eq:utils}
\util_i(\type_i) = \min\lbr{\underline{u}_i +\ability(\type_i-\underline{\type}_i),\,
\inf_{\type'_i\leq \type_i} \alloc_i(\type'_i)+\ability(\type_i-\type'_i)}.
\end{align}
\sloppy
Notice that $\util_i(\type_i)\leq \alloc_i(\type_i)$, because $\inf_{\type'_i\leq \type_i} \alloc_i(\type'_i)+\ability(\type_i-\type'_i)\leq \alloc_i(\type_i)$ for all $\type_i$.
For those types~$\type_i$ such that $\util_i(\type_i)<\alloc_i(\type_i)$, by the definition of $\util_i$, there exists some $\type_i'<\type_i$ such that $\util_i(\type_i) = \min\lbr{\underline{u}_i +\ability(\type_i-\underline{\type}_i),\,
 \alloc_i(\type'_i)+\ability(\type_i-\type'_i)}$, implying that $\util_i'(\type_i)=\ability$.
 For those types $\type_i$ such that $\util_i(\type_i)=\alloc_i(\type_i)$, by definition, $\alloc_i(\type'_i)+\ability(\type_i-\type'_i)\geq\alloc_i(\type_i)$ for all $\type_i'<\type_i$.
 This can happen only when $\alloc_i'(\type_i)<\ability$, implying that $\util_i'(\type_i)<\ability$.
Hence $(\allocs,\utils)$ satisfies \eqref{eq:IC} and is implementable by a non-coordination mechanism. 

Next we show that $\utils$ is the unique utility profile such that $(\allocs,\utils)$ is implementable by a non-coordination mechanism, given utilities $\{\underline{u}_i\}_{i=1,\dots,n}$ for the lowest types.
Suppose $\utils\primed$ is a different utility profile such that $(\allocs,\utils\primed)$ is implementable by a non-coordination mechanism
and $\util\primed_i(\underline{\type}_i) = \underline{u}_i$ for all $i$.
Then \eqref{eq:IC} implies that $\util\primed_i(\type_i) \leq \alloc_i(\type_i)$ for all $\type_i$.

Suppose there exists $\type_i$ such that $\util\primed_i(\type_i) > \util_i(\type_i)$. 
This is only possible if $\util_i(\type_i)<\alloc_i(\type_i)$ and  there exists some $\type'_i < \type_i$ such that $\util_i(\type_i)=\alloc_i(\type'_i)+\ability(\type_i-\type'_i) < \util\primed_i(\type_i)$.
However, this implies that in the direct mechanism $(\allocs,\utils\primed)$, agent $i$ with type $\type'_i$ has an incentive to misreport his type as $\type_i$. 
This contradicts the assumption that $(\allocs,\utils\primed)$ is implementable by a non-coordination mechanism.

Now suppose there exists $\type_i$ such $\util\primed_i(\type_i) < \util_i(\type_i)$.
This is only possible if $\util\primed_i(\type_i)<\alloc_i(\type_i)$. 
Let  $\type'_i= \underline{\type}_i$ if $\util\primed_i(\type_i) < \alloc_i(\type_i)$ for all $\type_i$. Otherwise, let $\type'_i=\sup\{z\leq \type_i:\util\primed_i(\type'_i) = \alloc_i(\type'_i)\}$.
In both cases, by \eqref{eq:IC}, we have $\util\primed_i(\type_i) = \util\primed_i(\type'_i) + \ability(\type_i-\type'_i)$.
In the case where $\type'_i= \underline{\type}_i$, we have 
\begin{align*}
\util\primed_i(\type_i) = \util\primed_i(\type'_i) + \ability(\type_i-\type'_i)<\util_i(\type_i)\leq \underline{u}_i+\ability(\type_i-\type'_i),
\end{align*}
implying that $\util\primed_i(\underline{\type}_i) < \underline{u}_i$, a contradiction. 
In the case where $\type'_i> \underline{\type}_i$, we can similarly infer that $\util\primed_i(\type'_i) < \alloc_i(\type'_i)$, which is again a contradiction. 
Hence, for any interim allocation rule $\allocs$, if there exists an interim utility $\utils$ such that $(\allocs,\utils)$ is implementable by a non-coordination mechanism, then $\utils$ is uniquely pinned down by $\allocs$ and the utility profile for the lowest types, and it is given by the expression \eqref{eq:utils}.

Finally, for any $\utils\primed$ such that $(\allocs,\utils\primed)$ is implementable by a non-coordination mechanism,
if $\util\primed_i(\underline{\type}_i) <  \util_i(\underline{\type}_i)$ for all~$i$, 
then by \eqref{eq:utils} 
we must have $\util\primed_i(\type_i) \leq \util_i(\type_i)$ for all $\type_i$.
\end{proof}

\begin{proof}[Proof of \cref{lmm:symmetry}]
Consider a relaxed problem $(\cP'_{\alpha})$ where, instead of the \eqref{eq:IC} constraints, we only require that 
$\util'_i(\type_i)\in [0, \ability]$
and $\util_i(\type_i)\leq \alloc_i(\type_i)$ for any agent $i$ with type $\type_i$.
Note that this is a convex constraint, 
and hence the relaxed problem is a convex problem. 
Thus, there exists a symmetric optimal solution $(\allocs,\utils)$ for Problem $(\cP'_{\alpha})$ if the environment is symmetric. 
Moreover, as $\utils$ is maximized given the derivative constraint and the upper bound of $\allocs$, 
the allocation--utility pair $(\allocs,\utils)$ also satisfies the \eqref{eq:IC} constraints
by the proof of \cref{cor:payoff equiv}. 
Therefore, $(\allocs,\utils)$ is also feasible 
and hence is an optimal solution for Problem \eqref{eq:P_alpha single}.
\end{proof}

\subsection{Characterization of the optimal mechanism.} To simplify the notation in the later analysis, given the partition of the type space, we add a degenerate interval $\underline{\type}^{(0)}=\bar{\type}^{(0)}=\bar{\type}$.

\begin{proof}[Proof of \cref{thm:optimal characterization}]
By \cref{lem:monotone alloc-util implementation}, the optimal utility function $\optimalutil$ must be continuous with subgradient between $0$ and~$\ability$,  
and $\optimalalloc(\type) = \optimalutil(\type)$ if $\optimalalloc'(\type)< \ability$. 
Therefore, we can partition the type space into countably many disjoint intervals $\{(\underline{\type}^{(j)}, \bar{\type}^{(j)})\}_{j=1}^{\infty}$, each of which falls into one of the following three categories:
\begin{enumerate}[{Case} 1:]
    \item \label{case1} 
    $\optimalalloc(\type)=\optimalutil(\type)$ 
    and $\optimalutil'(\type) < \ability$ 
    for any type $\type\in (\underline{\type}^{(j)},\bar{\type}^{(j)})$.
    
    \item \label{case2} 
    $\optimalalloc(\type)=\optimalutil(\type)$ and $\optimalutil'(\type) = \ability$ 
    for any type $\type\in (\underline{\type}^{(j)},\bar{\type}^{(j)})$. 

    \item \label{case3}
    $\optimalalloc(\type) > \optimalutil(\type)$ and $\optimalutil'(\type) = \ability$ 
    for any type $\type\in (\underline{\type}^{(j)},\bar{\type}^{(j)})$.
\end{enumerate}

For any interim allocation rule $\alloc$, let $\intalloc(\type)=\int_{\type}^{\bar{\type}}\alloc(t)\dd\dist(t)$. Notice that $\intalloc(\type)$ is a continuous function.
\begin{lemma}
\label{lmm:constraint}
If $\alloc$ is optimal, then $\intalloc(\type)- \int_{\type}^{\bar{\type}}\efficientalloc(t)\dd\dist(t)<0$ implies 
\begin{enumerate}[(A)]
    \item $\util(\type)=\alloc(\type)$, and
    \item either $\util'(\type)=\ability$ or $\util'(\type)=0$.
\end{enumerate}
\end{lemma}

\begin{corollary}
\label{cor:case3}
If $(\alloc,\util)$ is optimal, then $\alloc(\type) > \util(\type)$ implies $\intalloc(\type)- \int_{\type}^{\bar{\type}}\efficientalloc(t)\dd\dist(t)=0$ and $\alloc(\type) = \efficientalloc(\type)$ almost everywhere.
\end{corollary}

\begin{corollary}
\label{cor:case1a}
If $(\alloc,\util)$ is optimal, then $0<\util'(\type)<\ability$ implies $\intalloc(\type)- \int_{\type}^{\bar{\type}}\efficientalloc(t)\dd\dist(t)=0$ and $\alloc(\type) = \efficientalloc(\type)$ almost everywhere.
\end{corollary}

Corollary \ref{cor:case3} implies that 
in Case \ref{case3}, 
we have $\optimalalloc(\type) = \efficientalloc(\type)>\optimalutil(\type)$. 
Thus Case \ref{case3} will correspond to the efficient region. 

The analysis of Case \ref{case1} is decomposed into two subcases:
\begin{enumerate}[{Case 1}a:]
    \item \label{case1a}
    $\optimalalloc(\type)=\optimalutil(\type)$ 
    and $\optimalutil'(\type) \in (0, \ability)$ 
    for any type $\type\in (\underline{\type}^{(j)},\bar{\type}^{(j)})$.
    \item \label{case1b} 
    $\optimalalloc(\type)=\optimalutil(\type)$ 
    and $\optimalutil'(\type) =0$ 
    for any type $\type\in (\underline{\type}^{(j)},\bar{\type}^{(j)})$.
\end{enumerate}
By Corollary \ref{cor:case1a}, in Case 1\ref{case1a}, 
we have 
$\intalloc(\type)- \int_{\type}^{\bar{\type}}\efficientalloc(t)\dd\dist(t)=0$
and $\optimalalloc(\type) = \efficientalloc(\type)$
for any type $\type\in (\underline{\type}^{(j)},\bar{\type}^{(j)})$. 
Therefore, Case 1\ref{case1a} corresponds to the no-tension region. 
Moreover, we show that Case 1\ref{case1b} cannot occur (the proof is deferred to the end of the section). 
Thus Case \ref{case1} gives the no-tension region. 
\begin{lemma}\label{lem:no case 1b}
Case 1\ref{case1b} does not occur in the optimal solution. 
\end{lemma}

Finally, for any interval $j$ that corresponds to Case \ref{case2}, 
if $\underline{\type}^{(j)} > \underline{\type}$, 
since the integration constraint \eqref{eq:feasibility symmetric} binds for all types within each interval under any of the two other cases, 
it must also bind for both endpoints of the interval $j$; hence
\begin{align*}
\int_{\underline{\type}^{(j)}}^{\bar{\type}^{(j)}} \optimalalloc(\type) \dd \dist(\type) 
= \int_{\underline{\type}^{(j)}}^{\bar{\type}^{(j)}} \efficientalloc(\type) \dd \dist(\type).
\end{align*}
If $\underline{\type}^{(j)} = \underline{\type}$, 
then \eqref{eq:feasibility symmetric} also binds at $\underline{\type}$, 
since otherwise we could increase the allocation and utility for a sufficiently small region above type $\underline{\type}$
without violating feasibility, 
which would contradict the optimality of the solution. 
Hence the above equality again holds, 
and Case \ref{case2} corresponds to the no-effort region. 
\end{proof}

\begin{proof}[Proof of \cref{lmm:constraint}]
Consider the following relaxation of Problem \eqref{eq:P_alpha single}:
\begin{equation}
\tag{$\hat{\mathcal{R}}_{\alpha}$}
\label{eq:R_alpha single}
    \begin{aligned}
      \sup_{\alloc,\util} \quad& \expect[\type]{\alpha \cdot \type\cdot \alloc(\type)
     + (1-\alpha)\cdot \util(\type)} \\
     \text{s.t.} \quad& \int_{\type}^{\bar{\type}} \alloc(\type)\dd \dist(z) \leq \int_{\type}^{\bar{\type}} \efficientalloc(z) \dd\dist(z) \quad\forall \type\in[\underline{\type},\bar{\type}],\\
     & \util(\type) \leq \alloc(\type), \qquad 0\leq \util'(\type)\leq \ability.\\
    \end{aligned}
\end{equation}
Here we have omitted the monotonicity constraint on the allocation, as the optimal solution for the above program will automatically satisfy the monotonicity constraint. 

Define $\intalloc(\type)=\int_{\type}^{\bar{\type}}\alloc(t)\dd\dist(t)$ 
and $\intalloc'(\type)=-\alloc(\type)\density(\type)$.
The relaxed problem can be rewritten as follows:
\begin{align*}
\sup_{\intalloc,\util}\quad&\int_{\underline{\type}}^{\bar{\type}}
-\alpha\cdot\type \cdot \intalloc'(\type)+(1-\alpha)\cdot\util(\type)\cdot\density(\type)\dd\type
&\,\\
\text{s.t.} \quad& \intalloc(\type)\leq \int_{\type}^{\bar{\type}}\efficientalloc(t)\dd\dist(t),
&\hfill\lambda(\type),\\
&\util(\type)\density(\type)+\intalloc'(\type)\leq 0,
&\hfill\gamma(\type),\\
& 0\leq \util'(\type),
&\hfill\kappa_1(\type),\\
& \util'(\type)\leq \ability,
&\hfill\kappa_2(\type).
\end{align*}
The Lagrange multipliers $\lambda(\type),\gamma(\type), \kappa_1(\type), \kappa_2(\type)$ are non-negative.
The Lagrangian is given by
\begin{equation*}
\begin{aligned}
\hat{\mathcal{L}}(\intalloc,\intalloc',\util,\util',\lambda,\gamma,\kappa_1,\kappa_2)
=&\,-[\alpha\type\cdot \intalloc'(\type)-(1-\alpha)\cdot\util(\type)\density(\type)\\
&+\lambda(\type)(\intalloc(\type)- \int_{\type}^{\bar{\type}}\efficientalloc(t)\dd\dist(t))\\
&+\gamma(\type)(\util(\type)\density(\type)+\intalloc'(\type))\\
&+\kappa_1(\type)(\util'(\type) - \ability) -\kappa_2(\type)\util'(\type)].
\end{aligned}
\end{equation*}
The solution of the problem satisfies the following conditions: 
\begin{enumerate}[(1)]
\item The Euler--Lagrange conditions,\footnote{We are looking for piecewise continuous solutions (the state variables are continuous and the control variables are piecewise continuous), since, in principle, the allocation $\alloc(\type)$ may be merely piecewise continuous and not continuous, while $\util(\type)$ is continuous but its derivative might not be. The necessary conditions should be the integral form of the Euler--Lagrange conditions, together with the Erdmann--Weierstrass corner conditions \citep[cf.][]{clarke2013functional}. However, the latter have no bite here, and we can use the usual form of the Euler--Lagrange conditions, since they do not involve the state variables or the controls. Notice, though, that the Lagrange multiplier $\gamma(\type)$ could potentially be $PC^1$.} 
\begin{equation}
\tag{EL-1}\label{eq:EL-1}
\begin{aligned}
& \frac{\partial \hat{\mathcal{L}}}{\partial \intalloc}-\frac{\dd}{\dd\type}\frac{\partial \hat{\mathcal{L}}}{\partial \intalloc'}=0\Leftrightarrow & \lambda(\type)-(\alpha+\gamma'(\type))=0\\
\end{aligned}
\end{equation}
and
\begin{equation}
\tag{EL-2}\label{eq:EL-2}
\begin{aligned}
& \frac{\partial \hat{\mathcal{L}}}{\partial \util}-\frac{\dd}{\dd\type}\frac{\partial \hat{\mathcal{L}}}{\partial \util'}=0 \Leftrightarrow & (\gamma(\type)-(1-\alpha))\density(\type)-\kappa'(\type)=0,\\
\end{aligned}
\end{equation}
where $\kappa(\type)=\kappa_1(\type)-\kappa_2(\type)$, hold whenever they are well-defined.

\item The complementary slackness conditions hold:
\begin{align}
&\lambda(\type)(\intalloc(\type)- \int_{\type}^{\bar{\type}}\efficientalloc(t)\dd\dist(t))=0,\quad \lambda(\type)\geq 0,\tag{CS-1a} \label{eq:CS-1a}\\
&\gamma(\type)[\util(\type)\density(\type)+\intalloc'(\type)]=0,\quad \gamma(\type)\geq 0,\tag{CS-1b} \label{eq:CS-1b}\\
&\kappa_1(\type)[\util'(\type) - \ability]=0,\quad \kappa_1(\type)\geq 0,\tag{CS-1c} \label{eq:CS-1c}\\
&\kappa_2(\type)\util'(\type) =0,\quad \kappa_2(\type)\geq 0.\tag{CS-1d} \label{eq:CS-1d}
\end{align}
\end{enumerate}

Suppose $\intalloc(\type)- \int_{\type}^{\bar{\type}}\efficientalloc(t)\dd\dist(t)<0$.
We show that the following two conditions hold for the optimal solution:
\begin{itemize}
\item $\util(\type)=\alloc(\type)$. (By \eqref{eq:CS-1a}, $\lambda(\type)=0$ holds in an interval. 
From \eqref{eq:EL-1}, we have $\gamma'(\type)=-\alpha$.
Hence $\gamma(\type)$ cannot be a constant in this interval; in particular, $\gamma(\type)\neq 0$ except for at most one point.
Combined with \eqref{eq:CS-1b}, this further implies that $\util(\type)\density(\type)+\intalloc'(\type)=0$, i.e., $\util(\type)=\alloc(\type)$.)

\item Either $\util'(\type)=\ability$ or $\util'(\type)=0$.
(Reasoning similarly as for the previous condition, we have that $\gamma(\type)\neq 1-\alpha$ except for at most one point, which combined with \eqref{eq:EL-2} implies that $\kappa'(\type)\neq 0$ except for at most one point.
This means that $\kappa(\type)$ is not a constant; in particular, it is not zero.
The result follows from applying \eqref{eq:CS-1c} and \eqref{eq:CS-1d}.)\qedhere
\end{itemize}
\end{proof}

\begin{proof}[Proof of \cref{cor:case3}]
    The contrapositive of Lemma \ref{lmm:constraint} is also true:  $\alloc(\type) > \util(\type)$ implies $\intalloc(\type)- \int_{\type}^{\bar{\type}}\efficientalloc(t)\dd\dist(t)=0$.
    By rearranging the terms and taking the derivative with respect to~$\type$, 
    we have $\alloc(\type) = \efficientalloc(\type)$ almost everywhere. 
\end{proof}

\begin{proof}[Proof of \cref{lem:no case 1b}]
Suppose Case 1\ref{case1b} occurs. In this case, 
since $\util$ is a continuous function, 
$\optimalalloc(\type)=\optimalutil(\type) = \optimalutil(\bar{\type}^{(j)})$ for any type $\type\in (\underline{\type}^{(j)},\bar{\type}^{(j)})$.
Let $j'$ be the index of the interval such that $\bar{\type}^{(j)} = \underline{\type}^{(j')}$.
We consider three possible situations for interval $j'$:

\begin{itemize}
\item Interval $j'$ belongs to Case \ref{case1}. 
In this case, the integration constraint \eqref{eq:feasibility symmetric} binds at $\underline{\type}^{(j')}$,
and $\util(\bar{\type}^{(j)}) = \util(\underline{\type}^{(j')}) = \efficientalloc(\underline{\type}^{(j')})$.
Therefore, there exists a constant $\epsilon> 0$ such that \eqref{eq:feasibility symmetric} is violated at type $\bar{\type}^{(j)}-\epsilon$, a contradiction. 

\item Interval $j'$ belongs to Case \ref{case2}. 
In this case, \eqref{eq:feasibility symmetric} does not bind at type $\underline{\type}^{(j')}$.
Suppose otherwise; then we must have $\efficientalloc(\underline{\type}^{(j')}) \leq  \util(\bar{\type}^{(j)})$ 
in order for \eqref{eq:feasibility symmetric} to hold for type $\underline{\type}^{(j')}+\epsilon$ given sufficiently small $\epsilon>0$.
However, this would imply that \eqref{eq:feasibility symmetric} is violated for type $\bar{\type}^{(j)}-\epsilon$ given sufficiently small $\epsilon>0$.

Next we consider two cases for interval $j$.
\begin{itemize}
\item $\underline{\type}^{(j)} > \underline{\type}$.
In this case, \eqref{eq:feasibility symmetric} cannot bind at any type $\type\in [\underline{\type}^{(j)},\bar{\type}^{(j)})$. 
This is because if it binds at $\type$, then $\alloc(\type) = \util(\type) > \efficientalloc(\type)$. 
By the continuity of $\util$ and $\efficientalloc$, 
and the constraint that $\alloc \geq \util$, 
there exists a constant $\epsilon>0$ such that \eqref{eq:feasibility symmetric} is violated at $\type-\epsilon$. 
Thus, there exist $\epsilon,\delta>0$
such that for any type $\type\in[\underline{\type}^{(j)}-\epsilon, \bar{\type}^{(j)}+\epsilon]$, 
\begin{align*}
\int_{\type}^{\bar{\type}} \alloc(z)\dd \dist(z) \leq \int_{\type}^{\bar{\type}} \efficientalloc(z) \dd\dist(z) - \delta.
\end{align*}
Moreover, we can select $\epsilon$ to be sufficiently small to satisfy the additional condition that $\alloc'(\type)\leq\ability$ for any type $\type\in[\underline{\type}^{(j)}-\epsilon, \bar{\type}^{(j)}+\epsilon]$.
Given a parameter $\type^*$, let $\alloc\dprimed$ be the allocation such that 
\begin{enumerate}[(1)]
    \item $\alloc\dprimed(\type) = \alloc(\underline{\type}^{(j)}-\epsilon)$ 
    for any type $\type\in [\underline{\type}^{(j)}-\epsilon, \type^*]$;
    
    \item $\alloc\dprimed(\type) =\alloc(\underline{\type}^{(j)}-\epsilon) + \ability\cdot(\type-\type^*)$
    for any type $\type\in (\type^*, \type^* + \frac{1}{\ability}\cdot\alloc(\bar{\type}^{(j)}+\epsilon) - \alloc(\underline{\type}^{(j)}-\epsilon))$;
    
    \item $\alloc\dprimed(\type) = \alloc(\bar{\type}^{(j)}+\epsilon)$ 
    for any type $\type\in [\type^* +\frac{1}{\ability}\cdot\alloc(\bar{\type}^{(j)}+\epsilon) - \alloc(\underline{\type}^{(j)}-\epsilon), \bar{\type}^{(j)}+\epsilon]$.
\end{enumerate}
The parameter $\type^*$ is chosen so that 
\begin{align*}
\int_{\underline{\type}^{(j)}-\epsilon}^{\bar{\type}^{(j)}+\epsilon} \alloc\dprimed(z) \dd\dist(z)
=\int_{\underline{\type}^{(j)}-\epsilon}^{\bar{\type}^{(j)}+\epsilon} \alloc(z) \dd\dist(z).
\end{align*}
It is easy to verify that 
\begin{align*}
\int_{\underline{\type}^{(j)}-\epsilon}^{\bar{\type}^{(j)}+\epsilon} z\cdot\alloc\dprimed(z) \dd\dist(z)
>\int_{\underline{\type}^{(j)}-\epsilon}^{\bar{\type}^{(j)}+\epsilon} z\cdot\alloc(z) \dd\dist(z),
\end{align*}
since $\alloc\dprimed$ shifts allocation probabilities from low types to high types compared to $\alloc$. 
Therefore, given a sufficiently small constant $\hat{\delta}> 0$,
consider another allocation--utility pair $(\alloc\primed,\util\primed)$
such that 
\begin{enumerate}[(1)]
    \item $\alloc\primed(\type)=\alloc(\type)$ and $\util\primed(\type)=\util(\type)$ for any type $\type\not\in[\underline{\type}^{(j)}-\epsilon, \bar{\type}^{(j)}+\epsilon]$;
    \item $\alloc\primed(\type)=(1-\hat{\delta})\cdot\alloc(\type)+\hat\delta\cdot \alloc\dprimed(\type)$ and $\util\primed(\type)=(1-\hat{\delta})\cdot\util(\type)+\hat\delta\cdot \alloc\dprimed(\type)$ for any type $\type\in[\underline{\type}^{(j)}-\epsilon, \bar{\type}^{(j)}+\epsilon]$.
\end{enumerate}
The new allocation--utility pair $(\alloc\primed,\util\primed)$ is feasible and strictly improves the objective value, 
a contradiction to the optimality of $(\alloc,\util)$.

\item $\underline{\type}^{(j)} = \underline{\type}$. 
The proof for this case is similar. The only difference is that we can change the allocation and utility within interval $j$ without worrying about the continuity of the utility function for lower types. 
Therefore, using a similar construction for $\alloc\dprimed$ and $(\alloc\primed,\util\primed)$, restricted to the interval 
$[\underline{\type}^{(j)}, \bar{\type}^{(j)}+\epsilon]$
for sufficiently small $\epsilon>0$, 
we can again show that the allocation--utility pair $(\alloc,\util)$ that contains Case 1\ref{case1b} is not optimal.
\end{itemize}

\item Either interval $j'$ belongs to Case \ref{case3}, 
or $\underline{\type}^{(j')}$ is the highest possible type $\bar{\type}$. 
In either case, for the integration constraint \eqref{eq:feasibility symmetric} to be satisfied within interval $j$, both the efficient allocation $\efficientalloc$ and the interim allocation $\alloc$ must be strictly above the utility at~$\underline{\type}^{(j')}$.
Therefore, the allocation within interval $j$ can be increased, relative to allocations above $\underline{\type}^{(j')}$, without violating the monotonicity.
Again we use a similar construction for $\alloc\dprimed$ and $(\alloc\primed,\util\primed)$, restricted to the interval 
$[\underline{\type}^{(j)}-\epsilon, \bar{\type}^{(j)}]$
for sufficiently small $\epsilon>0$.
Here we add the further operation of increasing the utility $\util\primed$ for types above $\underline{\type}^{(j')}$ 
to maintain the monotonicity of the utility function; this only increases the objective value. 
Thus, the allocation--utility pair $(\alloc,\util)$ that contains Case 1\ref{case1b} is not optimal.\qedhere
\end{itemize}
\end{proof}

\section{Proof of \cref{thm:scarce_resource}}
\label{apx:A5_scarce}

\begin{proof}[Proof of \cref{lem:large_convex}]
Taking the second-order derivative gives us
\begin{align*}
\efficientalloc'' &= (\dist^{n-1})'' = ((n-1) \dist^{n-2}\cdot\density)'
= (n-1)((n-2)\dist^{n-3} \cdot \density^2 + \dist^{n-2}\cdot \density')\\
&\geq (n-1)\dist^{n-3}((n-2)\underline{\beta}_1^2 - \dist\cdot\beta_2)
\geq 0
\end{align*}
when $n\geq N \geq 2+\frac{\beta_2}{\underline{\beta}_1^2}$.
\end{proof}

\begin{proof}[Proof of \cref{prop:convex optimal}]
\sloppy 
By \cref{thm:optimal characterization}, there exists a partition of the type space 
$\{(\underline{\type}^{(j)}, \bar{\type}^{(j)})\}_{j=1}^{\infty}$
such that each interval belongs to one of the three cases.
It is sufficient to show that the order of the three cases on the type space cannot be changed in the optimal non-coordination mechanism. 

First we show that for $j$ such that interval $j$ is in the no-tension region, 
it is optimal for all intervals containing types below $\underline{\type}^{(j)}$ to be in the no-tension region as well. 
The main reason is that by the convexity of the efficient allocation rule, 
for any type $\type\leq \underline{\type}^{(j)}$, 
$\efficientalloc'(\type) \leq \efficientalloc'(\underline{\type}^{(j)}) \leq\ability$. 
Therefore, if we set $\optimalutil(\type) = \optimalalloc(\type) = \efficientalloc(\type)$,
the resulting non-coordination mechanism is feasible and trivially maximizes the objective value. 

Let $\type^{(1)}$ be the supremum of the set of all types $\type$ lying in the no-tension region. 
The argument in previous paragraph shows that the whole interval $(\underline{\type},\type^{(1)})$ is in the no-tension region. 
Moreover, by \cref{thm:optimal characterization}, 
$\optimalalloc(\type^{(1)})=\optimalutil(\type^{(1)}) = \efficientalloc(\type^{(1)})$,
and for any $\type\geq \type^{(1)}$ we have $\optimalutil'(\type) = \ability$.
Now we consider two cases:
\begin{itemize}
\item If $\efficientalloc'(\type^{(1)}) \geq \ability$, 
then by the convexity of the efficient allocation rule, $\efficientalloc(\type) > \optimalutil(\type)$ for any type $\type > \type^{(1)}$, 
which implies that 
\begin{align*}
        \int_{\underline{\type}^{(j)}}^{\bar{\type}^{(j)}} \optimalutil(\type) \dd \dist(\type) 
        < \int_{\underline{\type}^{(j)}}^{\bar{\type}^{(j)}} \efficientalloc(\type) \dd \dist(\type)
\end{align*}
for any interval $j$ with types above $\type^{(1)}$. 
In this case, $\type^{(1)}=\type^{(2)}$ and the no-effort region does not exist.

\item If $\efficientalloc'(\type^{(1)}) < \ability$, then 
$\efficientalloc(\type) < \optimalutil(\type)$ for any type $\type$ sufficiently close to $\type^{(1)}$. 
Therefore, for $j$ such that $\underline{\type}^{(j)} = \type^{(1)}$, interval $j$ must be in the no-effort region. 
Let $\type^{(2)} = \bar{\type}^{(j)}$. 
Note that for the integration constraint to be satisfied in interval $j$, 
we must have $\efficientalloc(\type^{(2)}) \geq \optimalutil(\type^{(2)})$
and $\efficientalloc'(\type^{(2)}) \geq \ability$.
Therefore, for any type $\type > \type^{(2)}$, we have $\efficientalloc(\type) > \optimalutil(\type)$; hence any interval above type $\type^{(2)}$ is in the efficient region.\qedhere
\end{itemize}
\end{proof}

\begin{proof}[Proof of \cref{thm:scarce_resource}]
By \cref{lem:large_convex}, for sufficiently large $n$, the efficient allocation rule is convex. 
Therefore, the interim allocation rule of the optimal non-coordination mechanism takes the form described in \cref{prop:convex optimal}.

Let $\optimalallocn{n}(\type)$ and the $\efficientallocn{n}(\type)$ 
be the optimal interim allocation rule in a non-coordination mechanism and efficient allocation rule respectively with $n<\infty$ agents.
For any finite $n$, we have that 
\begin{align*}
\frac{1}{n}\geq\int_{\type^{(1)}_n}^{\bar\type}\efficientallocn{n}(\type)\dd\dist(\type) 
&\geq \int_{\type^{(1)}_n}^{\bar\type} \rbr{\ability \cdot(\type-\type^{(1)}_n)+\efficientallocn{n}(\type^{(1)}_n)}\dd\dist(\type). 
\end{align*}
The first inequality holds because the ex-ante probability that a given agent gets the item is at most $\frac{1}{n}$, 
and the second inequality holds because the efficient allocation majorizes the interim allocation, since the latter is again at least the interim utility. 
Since $\efficientallocn{n}(\type^{(1)}_n)$ is non-negative, 
we have that 
\begin{align*}
\int_{\type^{(1)}_n}^{\bar\type} (\type-\type^{(1)}_n)\dd\dist(\type)
\leq \frac{1}{n\ability}
\end{align*}
for any $n$. 
Note that $\frac{1}{n\ability}$ converges to $0$ as $n$ approaches infinity. 
In order for the inequality to hold, 
$\type^{(1)}_n$ must also converge to $\bar{\type}$ as $n$ approaches infinity.
\end{proof}


\section{Proof of \cref{thm:scarce_nonconvergence}}
\begin{proof}[Proof of \cref{thm:scarce_nonconvergence}]
First we present \cref{lmm:efficientutil}, whose proof is given later in this section.
\Cref{lmm:efficientutil} says that given the efficient allocation rule, 
the sum of the expected utilities of the agents is small compared to the best scenario, i.e., the scenario in which the highest type gets the item without exerting effort, which is 1. 

\begin{lemma}\label{lmm:efficientutil}
For any $\epsilon > 0$, there exists $N_0\geq 1$ such that for any $n\geq N_0$, we have $n\cdot \expect[\type\sim \dist]{\efficientutiln{n}(\type)} \leq 1-\frac{1}{e}+\epsilon.$
\end{lemma}
Intuitively, this means that competition is high among agents with sufficiently high types. 
Thus agents with high types need to exert high effort to ensure a large allocation, leading to a utility loss relative to the first-best utility. 
By applying \cref{lmm:efficientutil},
we obtain an upper bound on the performance of the WTA contest. 
That is, for any $\epsilon>0$, there exists $N_0$ such that for any $n\geq N_0$, we have 
\begin{align*}
n\cdot V_{\alpha}(\efficientallocn{n}) 
&= n\alpha \cdot\expect[\type\sim \dist]{\type\cdot\efficientallocn{n}(\type)}
+n(1-\alpha)\cdot\expect[\type\sim \dist]{\efficientutiln{n}(\type)}\\
&\leq \alpha \cdot \bar{\type} + (1-\alpha)\cdot\rbr{1-\frac1e+\epsilon}.
\end{align*}
The inequality holds by \cref{lmm:efficientutil} and 
the fact that the upper bound on the type of the agent winning the item is $\bar{\type}$.

Next we provide a lower bound on the performance of the optimal contest. 
In particular, for any $n$ large enough, consider a feasible allocation 
\begin{align*}
    \alloc_n(\type)=\begin{cases}
    \begin{aligned}
        &\efficientallocn{n}(\type) &\text{ if } \type\leq \hat{\type}_{n}, \\
        & \ability\cdot(\type-\hat{\type}_n)+\efficientallocn{n}(\hat{\type}_n) &\text{ if } \type> \hat{\type}_{n}, \\
    \end{aligned}
    \end{cases}
\end{align*}
such that $\expect[\type\sim\dist]{\alloc_n(\type)}=\expect[\type\sim\dist]{\efficientallocn{n}(\type)}=\frac1n$.
Let $\util_n(\type)=\alloc_n(\type)$. 
Notice that $(\alloc_n,\util_n)$ satisfies the \eqref{eq:IC} constraints. 
Moreover, $\alloc_n(\type)$ induces no effort and 
hence $\expect[\type\sim \dist]{\util_n(\type)}=\frac1n$.
In the following lemma (proved at the end of this section), we show that the matching efficiency of the given allocation rule converges to the optimal welfare when the number of agents is sufficiently large. 
\begin{lemma}\label{lem:optutil}
For any $\epsilon > 0$, 
there exists $N_1$ such that for any $n\geq N_1$, 
$n\cdot\expect[\type\sim \dist]{\type\cdot\alloc_n(\type)}\geq \bar{\type}-\epsilon$.
\end{lemma}
Therefore, there exists $N_1$ such that for any $n\geq N_1$, 
we have
\begin{align*}
n\cdot V_{\alpha}(\optimalallocn{n}) 
& \geq n\cdot\alpha \expect[\type\sim \dist]{\type\cdot\alloc_n(\type)}+n\cdot(1-\alpha)\expect[\type\sim \dist]{\util_n(\type)}\\
& \geq \alpha (\bar{\type}-\epsilon)+1-\alpha.
\end{align*}

Finally, for any $\epsilon > 0$, letting $N=\max\lbr{N_0,N_1}$, 
we can combine the inequalities above to obtain 
\begin{equation*}
\frac{V_{\alpha}(\optimalallocn{n})}{V_{\alpha}(\efficientallocn{n})}
\geq \frac{(\bar{\type}-\epsilon)\cdot\alpha + 1 - \alpha}{\bar{\type}\cdot\alpha + (1-\alpha)(1-\frac{1}{e}+\epsilon)}
\end{equation*}
for any $n\geq N$.  \qedhere
\end{proof}

\begin{proof}[Proof of \cref{lmm:efficientutil}]
Let $n$ be a sufficiently large number so that $\efficientallocn{n}$ is convex,
and let $\type^{\dagger}_n$ be the cutoff type such that in the incentive-compatible implementation of efficient allocation, agents with any type $\type>\type^\dag_n$ exert costly effort, i.e., $\efficientallocn{n}'(\type^\dagger_n) = \ability$.
In other words, $(n-1) \cdot \dist^{n-2}(\type^\dagger_n)\cdot \density(\type^\dagger_n) = \ability$.
Rearranging the terms, we have 
\begin{align*}
\dist^{n-2}(\type^\dagger_n) = \frac{\ability}{(n-1) \cdot \density(\type^\dagger_n)}.
\end{align*}
Note that by \cref{asp:continuity}, the right-hand side is bounded below by $\frac{\ability}{(n-1) \cdot \bar{\beta}_1}$. 
Therefore, for any $\epsilon_0>0$, there exists $N_0$ such that for any $n\geq N_0$, we have 
\begin{align*}
\dist(\type^\dagger_n) \geq \rbr{\frac{\ability}{(n-1) \cdot \bar{\beta}_1}}^{\frac{1}{n-2}}
\geq 1-\epsilon_0.
\end{align*}
Since the density is bounded below by $\underline{\beta}_1$, 
we have that $\type^\dagger_n \geq \bar{\type} - \frac{\epsilon_0}{\underline{\beta}_1}$. 
For any $\epsilon_1 > 0$, let $N_1$ be an integer such that 
$\frac{\ability}{(n-1) \cdot \density(\type^\dagger_n)} \leq \epsilon_1$ for any $n\geq N_1$.
The expected utility of an agent with type $\bar{\type}$
is 
\begin{align*}
\efficientutiln{n}(\bar{\type})
= \dist^{n-1}(\type^\dagger_n) + \ability(\bar{\type} - \type^\dagger_n) \leq \dist^{n-2}(\type^\dagger_n)
+ \frac{\ability\cdot \epsilon_0}{\underline{\beta}_1}
\leq \epsilon_1 + \frac{\ability\cdot \epsilon_0}{\underline{\beta}_1}.
\end{align*}

Let $\type^\ddag_n$ be the type such that 
$\dist(\type^\ddag_n) = 1-\frac{1}{n}$. 
There exists $N_2$ such that $\type^\ddag_n \geq \type^\dag_n$
for any $n\geq N_2$. 
For any $\epsilon > 0$, let $\epsilon_1 = \frac{\epsilon}{2}$, 
$\epsilon_0 = \frac{\epsilon\underline{\beta}_1}{2\ability}$,
and $N = \max\{N_0,N_1,N_2\}$. For any $n\geq N$,
the expected effort of any agent is at least 
his effort from types above $\type^\ddag_n$, 
which is bounded below by 
\begin{align*}
(1-\dist(\type^\ddag_n)) \cdot (\efficientallocn{n}(\type^\ddag_n) - \efficientutiln{n}(\bar{\type}))
\geq \frac{1}{n}\rbr{\frac{1}{e} - \epsilon_1 - \frac{\ability\cdot \epsilon_0}{\underline{\beta}_1}}
= \frac{1}{n}\rbr{\frac{1}{e} - \epsilon}.
\end{align*}
Since the item is always allocated in equilibrium, the total utility is 
\begin{equation*}
n\cdot \expect[\type\sim \dist]{\efficientutiln{n}(\type)} 
\leq 1-\frac{1}{e}+\epsilon.\qedhere
\end{equation*}
\end{proof}

\begin{proof}[Proof of \cref{lem:optutil}]
Note that compared to the efficient allocation $\efficientallocn{n}$, 
the chosen allocation rule $\alloc_n$ only randomizes the allocation for types between $\hat{\type}_n$ and $\bar{\type}$. 
Therefore, we have 
\begin{align*}
n\cdot\expect[\type\sim \dist]{\type\cdot\alloc_n(\type)}
\geq n\cdot\expect[\type\sim \dist]{\type\cdot\efficientallocn{n}(\type)}
- (\bar{\type} - \hat{\type}_n).
\end{align*}
As in the proof of \cref{thm:scarce_resource}, we can show that $\lim_{n\rightarrow \infty} \hat{\type}_n=\bar{\type}$. 
By taking the limit of the above inequality, we have that
\begin{align*}
\lim_{n\rightarrow \infty} n\cdot\expect[\type\sim \dist]{\type\cdot\alloc_n(\type)}
\geq \lim_{n\rightarrow \infty} n\cdot\expect[\type\sim \dist]{\type\cdot\efficientallocn{n}(\type)}
= \bar{\type}.
\end{align*}
Thus, for any $\epsilon > 0$, 
there exists $N_1$ such that for any $n\geq N_1$, 
$n\cdot\expect[\type\sim \dist]{\type\cdot\alloc_n(\type)}\geq \bar{\type}-\epsilon$. 
\end{proof}

\section{Proof of \cref{thm:large_scale}}
\label{apx:A6_large_scale}
It is tempting to conjecture that when $z$ is large enough, $\efficientallocn{z}(\type)$ has an S shape (i.e., it is convex for small $\type$ and concave for large $\type$), which would naturally imply the order of the intervals as stated in our result. 
However, this is not true in general.\footnote{The second-order derivative of the allocation is 
\begin{align*}
\efficientallocn{z}''(\type) &= (zn-1)\cdot \binom{zn-2}{zk-1} (1-\dist(\type))^{zk-2}
\cdot (\dist(\type))^{z(n-k)-2}\cdot \\
&\qquad\rbr{\density^2(\type)(z(n-k)-1-(zn-2)\dist(\type))+\density'(\type)(1-\dist(\type))\dist(\type)}.
\end{align*}
No matter how large the parameter $z$ is, for types within $(\type_c-\epsilon_0,\type_c+\epsilon_0)$, the sign of the second-order derivative may change multiple times.}
To circumvent this inconvenience, note that for any small constant $\epsilon_0$,
when~$z$ is large enough, 
the interim efficient allocation has small slope (smaller than the marginal cost of effort $\ability$) outside the small interval $(\type_c-\epsilon_0,\type_c+\epsilon_0)$ centered at $\type_c$. 
Moreover, since the value of the efficient allocation changes a lot in this small interval, 
the agents will exert high effort in equilibrium if the items are allocated efficiently, 
leading to low expected utility for types around $\type_c$. 
We show that in the optimal contest, 
the principal randomizes the allocation around $\type_c$. 
In particular, the no-effort region, where the allocation is randomized, 
will cover the whole interval $(\type_c-\epsilon_0,\type_c+\epsilon_0)$. 
Since the derivative of the efficient allocation outside this region is at most $\ability$, 
the principal's objective value is maximized by the efficient allocation. 
We provide the formal proof below.

\begin{proof}[Proof of \cref{thm:large_scale}]
Since the distribution is continuous, the probability there is a tie for any two distinct types is $0$. 
Therefore, given the scale parameter $z$, the interim efficient allocation is 
\begin{align*}
\efficientallocn{z}(\type) = \prob[]{\type_{(nz-kz:nz-1)}\leq \type}
= \sum_{j=0}^{zk-1}\binom{zn-1}{j} 
\cdot (1-\dist(\type))^j\cdot (\dist(\type))^{zn-1-j},
\end{align*}
where $\type_{(nz-kz:nz-1)}$ is the $(nz-kz)$th order statistic, i.e., the $(nz-kz)$th smallest value in a sample of $nz-1$ observations, 
and the binomial coefficient $\binom{n}{k}$ is defined by 
$\binom{n}{k} = \frac{n!}{k!(n-k)!}$.

Recall that $\type_c$ is the cutoff type such that $1-\dist(\type_c) = \frac{k}{n}$. 
The derivative of the allocation is 
\begin{align*}
\efficientallocn{z}'(\type) = \density(\type)\cdot(zn-1)\cdot \binom{zn-2}{zk-1} (1-\dist(\type))^{zk-1}
\cdot (\dist(\type))^{z(n-k)-1}.
\end{align*}
Note that $\binom{zn-2}{zk-1} (1-\dist(\type))^{zk-1}
\cdot (\dist(\type))^{z(n-k)-1}$ is the probability that 
the binomial random variable $B(zn-2, 1-\dist(\type))$ equals $zk-1$.
When $1-\dist(\type) < \frac{k}{n}$, 
this probability becomes exponentially small as $zn$ increases, 
which implies that $\lim_{z\to\infty} \efficientallocn{z}'(\type) = 0$.
Therefore, for any $\epsilon_0 > 0$, 
there exists $Z_0$ such that for any $z\geq Z_0$, 
for any type $\type\not\in [\type_c-\epsilon_0, \type_c+\epsilon_0]$,
\begin{align*}
\efficientallocn{z}'(\type) \leq \ability.
\end{align*}
Again by Hoeffding's inequality, for any $\epsilon_1 > 0$, 
there exists $Z_1$ such that for any $z\geq Z_1$, 
\begin{align*}
\efficientallocn{z}(\type) \leq \epsilon_1
\end{align*}
for any type $\type\leq\type_c-\epsilon_0$
and 
\begin{align*}
\efficientallocn{z}(\type) \geq 1-\epsilon_1
\end{align*}
for any type $\type\geq\type_c+\epsilon_0$.
Intuitively, this is because $\lim_{z\rightarrow \infty}\efficientallocn{z}(\type)$ is a step function, i.e.,
\begin{equation*}
\lim_{z\rightarrow \infty}\efficientallocn{z}(\type)=\begin{cases}
    \begin{aligned}
       & 0 &\text{ if } \type<\type_c, \\
       & 1 &\text{ if } \type\geq\type_c.
    \end{aligned}
\end{cases}
\end{equation*}
Let $\tilde\type^{(1)} \triangleq \type_c - \epsilon_0 - \sqrt{\frac{8\epsilon_0\bar{\beta}_1}{\ability\underline{\beta}_1}}$. 
\begin{lemma}\label{lem:high_util}
For sufficiently large $z$, 
in the optimal contest $(\alloc_{\alpha,z},\util_{\alpha,z})$, we have 
$\util_{\alpha,z}(\tilde\type^{(1)}) > \efficientallocn{z}(\tilde\type^{(1)})$.
\end{lemma}
We defer the proof of the lemma to the end of this section. 
Note that in the optimal contest $(\alloc_{\alpha,z},\util_{\alpha,z})$, 
$\util_{\alpha,z}(\tilde\type^{(1)}) > \efficientallocn{z}(\tilde\type^{(1)})$
implies that type $\tilde\type^{(1)}$ must belong to a no-effort interval. 
Let $\type^{(1)}< \tilde\type^{(1)}< \type^{(2)}$ be the endpoints of this no-effort interval. 
Let $\Theta_+$ be the set of types in $(\type^{(1)},\type^{(2)})$
such that $\efficientallocn{z}(\type) > \hat\alloc_{\alpha,z}(\type)$,
and let $\Theta_-$ be the set of types in $(\type^{(1)},\type^{(2)})$
such that $\efficientallocn{z}(\type) < \alloc_{\alpha,z}(\type)$.
Since the integration constraint binds within $(\type^{(1)},\type^{(2)})$, 
we have that 
\begin{align*}
\int_{\Theta_+}(\efficientallocn{z}(\type)-\alloc_{\alpha,z}(\type))\dd\dist(\type)
+ \int_{\Theta_-}(\efficientallocn{z}(\type)-\alloc_{\alpha,z}(\type))\dd\dist(\type) = 0.
\end{align*}
Note that 
\begin{align*}
&\int_{\Theta_-}(\efficientallocn{z}(\type)-\alloc_{\alpha,z}(\type))\dd\dist(\type)
\leq -\int_{\type^{(1)}}^{\type_c-\epsilon_0}
(\ability(\type - \type^{(1)})-\epsilon_1) \dd\dist(\type)\\
&\leq -\underline{\beta}_1 \cdot
\rbr{\frac{\ability}{2}\cdot(\type_c-\epsilon_0-\type^{(1)})^2 - \epsilon_1\cdot(\type_c-\epsilon_0-\type^{(1)})}.
\end{align*}
Similarly, 
\begin{align*}
\int_{\Theta_+}(\efficientallocn{z}(\type)-\alloc_{\alpha,z}(\type))\dd\dist(\type)
\leq \int_{\type_c-\epsilon_0}^{\type^{(2)}}
1 \dd\dist(\type)
\leq \bar{\beta}_1 \cdot
(\type^{(2)} - \type_c+\epsilon_0).
\end{align*}
Combining the inequalities above, for sufficiently small $\epsilon_1\leq \frac{\ability}{4}(\type_c-\epsilon_0-\type^{(1)})$, 
we must have 
\begin{align*}
{\type}^{(2)}\geq \type_c-\epsilon_0 
+ \frac{\ability\cdot\underline{\beta}_1}{4\bar{\beta}_1} \cdot(\type_c-\epsilon_0-\type^{(1)})^2
\geq \type_c+\epsilon_0. 
\end{align*}
We obtain the last inequality simply by substituting the bound for $\type^{(1)}$.
This implies that in the optimal contest, the no-effort region 
$(\type^{(1)},\type^{(2)})$ covers the whole interval $(\type_c-\epsilon_0,\type_c+\epsilon_0)$.
Note that since the derivative of the efficient allocation outside the no-effort region $(\type^{(1)},\type^{(2)})$ is at most $\ability$, 
the principal's objective is maximized by the efficient allocation. 
In particular, let $\type^{(3)}\geq \type^{(2)}$ be the type such that the linear extension of the utility function within the no-tension region intersects the efficient allocation rule. 
Then the interval $(\type^{(2)},\type^{(3)})$ is the efficient region, 
and the union of $(\underline{\type},\type^{(1)})$ 
and $(\type^{(3)},\bar{\type})$ is the no-tension region.
\end{proof}

\begin{proof}[Proof of \cref{lem:high_util}]
It is sufficient to show that any contest $(\tilde{\alloc}_{\alpha,z},\tilde{\util}_{\alpha,z})$
such that $\tilde{\util}_{\alpha,z}(\tilde\type^{(1)}) \leq \efficientallocn{z}(\tilde\type^{(1)})$
cannot be an optimal contest. 
We prove this by contradiction: given such a contest, we construct a contest $\hat{\alloc}_{\alpha,z},\hat{\util}_{\alpha,z}$ that yields a higher objective value.

Let $\epsilon_0,\epsilon_1,\epsilon_2>0$ be any numbers such that the following hold:\footnote{Notice that these inequalities can hold at the same time: if one chooses $\epsilon_0$ and $\epsilon_1$ that are ``small'' compared to $\epsilon_2$, for example, $\epsilon_0=o(\epsilon_2^4)$ and $\epsilon_1=o(\epsilon_2^4)$, then the last inequality holds because the left-hand side is of higher order than the right-hand side.}
\begin{align*}
    & 0<\epsilon_0\leq \min\lbr{\frac{\underline{\beta}_1}{10\ability\cdot\bar{\beta}_1},\epsilon_2^4}, \qquad\epsilon_0+2\sqrt{\frac{8\epsilon_0\bar{\beta}_1}{\ability\underline{\beta}_1}}\leq \epsilon_2, \\
   & 0<\epsilon_1\leq \min\{0.01,\epsilon_2^4\}, \qquad 0<\epsilon_2<\frac{\underline{\beta}_1}{10\ability\cdot\bar{\beta}_1},\\
   &\alpha\bar{\beta}_1 \cdot\rbr{(\epsilon_2+\epsilon_0)^2\cdot\frac{\bar{\beta}_1}{\underline{\beta}_1}+\epsilon_0+\epsilon_2}^2
< \frac{1}{2\ability}(1-\alpha) \cdot \rbr{\ability\rbr{\epsilon_2 - \epsilon_0- \sqrt{\frac{8\epsilon_0\bar{\beta}_1}{\ability\underline{\beta}_1}}} - \epsilon_1}. 
\end{align*}
Let $\hat\type^{(1)}\triangleq \type_c - \epsilon_2$. 
By our choice of $\epsilon_0$,
we have $\hat\type^{(1)} < \tilde\type^{(1)}$. 

Consider a contest  $(\hat{\alloc}_{\alpha,z},\hat{\util}_{\alpha,z})$ characterized by three cutoffs $\hat{\type}^{(1)}<\hat{\type}^{(2)}\leq\hat{\type}^{(3)}$ 
such that 
the union of $(\underline{\type}, \hat{\type}^{(1)})$ and $(\hat{\type}^{(3)},\bar{\type})$ is the no-tension region, 
$(\hat{\type}^{(1)},\hat{\type}^{(2)})$ is the no-effort region, and 
$(\hat{\type}^{(2)},\hat{\type}^{(3)})$ is the efficient region.

\noindent\textbf{Step 1:} In this step, we will show that if $\hat\type^{(1)}$ is chosen so that $\type_c-  \frac{\underline{\beta}_1}{10\ability\cdot\bar{\beta}_1}\leq \hat\type^{(1)}$,\footnote{Such a choice is possible because by the choice of $\epsilon_0,\epsilon_1,\epsilon_2$, we have  $\type_c-  \frac{\underline{\beta}_1}{10\ability\cdot\bar{\beta}_1}\leq \tilde\type^{(1)}$.} then the integration constraint for the no-effort interval imposes an upper bound on the length of the no-effort interval, i.e., $\hat{\type}^{(2)}\leq \tilde{\type}$, where $\tilde{\type}\triangleq \type_c+\epsilon_0 + \frac{2\ability\cdot\bar{\beta}_1}{\underline{\beta}_1}\cdot(\type_c+\epsilon_0-\hat\type^{(1)})^2$.

Let $\hat{\Theta}_+$ be the set of types in $(\hat\type^{(1)},\hat\type^{(2)})$
such that $\efficientallocn{z}(\type) > \hat\alloc_{\alpha,z}(\type)$,
and let $\hat{\Theta}_-$ be the set of types in $(\hat\type^{(1)},\hat\type^{(2)})$
such that $\efficientallocn{z}(\type) < \hat\alloc_{\alpha,z}(\type)$.
Since the integration constraint binds within $(\hat\type^{(1)},\hat\type^{(2)})$, 
we have that 
\begin{align*}
0=\,&\int_{\hat{\Theta}_+}(\efficientallocn{z}(\type)-\hat\alloc_{\alpha,z}(\type))\dd\dist(\type)
+ \int_{\hat{\Theta}_-}(\efficientallocn{z}(\type)-\hat\alloc_{\alpha,z}(\type))\dd\dist(\type)\\
\geq\,&\int_{\type_c+\epsilon_0}^{\hat\type^{(2)}}
(1-2\epsilon_1 - \ability(\type - \hat\type^{(1)})) \dd\dist(\type)
-\int_{\hat\type^{(1)}}^{\type_c+\epsilon_0}
\ability(\type - \hat\type^{(1)}) \dd\dist(\type).
\end{align*}
 
By our choice of $\hat\type^{(1)}$ and $\epsilon_0, \epsilon_1$, we have 
 $1-2\epsilon_1 - \ability(\type - \hat\type^{(1)})\geq \frac{1}{2}$ for any type $\type\leq \tilde{\type}$. 
Therefore,
\begin{align*}
&\int_{\type_c+\epsilon_0}^{\tilde{\type}}
(1-2\epsilon_1 - \ability(\type - \hat\type^{(1)})) \dd\dist(\type)
-\int_{\hat\type^{(1)}}^{\type_c+\epsilon_0}
\ability(\type - \hat\type^{(1)}) \dd\dist(\type)\\
\geq\,& \frac{\underline{\beta}_1}{2}(\tilde{\type}-\type_c-\epsilon_0)
- \ability\cdot \bar{\beta}_1 (\type_c+\epsilon_0-\hat\type^{(1)})^2
\geq 0.
\end{align*}
Combining the above two inequalities, we get the desired bound on $\hat{\type}^{(2)}$.

\noindent\textbf{Step 2:} Next we utilize the upper bound to show that the objective value from the contest $\hat{\alloc}_{\alpha,z},\hat{\util}_{\alpha,z}$ is higher than that from the contest $\tilde{\alloc}_{\alpha,z},\tilde{\util}_{\alpha,z}$ with $\tilde{\util}_{\alpha,z}(\tilde\type^{(1)})\leq \tilde{\alloc}_{\alpha,z}(\tilde\type^{(1)})$.
Note that $\efficientallocn{z}$ and $\hat\alloc_{\alpha,z}(\type)$ coincide at any type $\type$ outside the no-effort region. 
Therefore, the loss in efficiency compared to the efficient allocation rule is 
\begin{align*}
& \alpha \cdot \int_{\hat\type^{(1)}}^{\hat\type^{(2)}}\type\cdot \efficientallocn{z}\dd\dist(\type) 
- \alpha \cdot \int_{\hat\type^{(1)}}^{\hat\type^{(2)}}\type\cdot\hat{\alloc}_{\alpha,z}(\type)\dd\dist(\type)\\
=\,& \alpha \cdot\int_{\hat{\Theta}_+}\type\cdot(\efficientallocn{z}(\type)-\hat\alloc_{\alpha,z}(\type))\dd\dist(\type)
-\alpha \cdot\int_{\hat{\Theta}_-}\type\cdot(\efficientallocn{z}(\type)-\hat\alloc_{\alpha,z}(\type))\dd\dist(\type)\\
\leq\,& \alpha \cdot(\hat\type^{(2)} - \hat\type^{(1)})\cdot
\int_{\hat{\Theta}_+}(\efficientallocn{z}(\type)-\hat\alloc_{\alpha,z}(\type))\dd\dist(\type)\\
\leq\,& \alpha \cdot(\hat\type^{(2)}-\hat\type^{(1)})\cdot
(\dist(\hat\type^{(2)}) - \dist(\hat\type^{(1)}))
\leq \alpha\bar{\beta}_1 \cdot(\hat\type^{(2)}-\hat\type^{(1)})^2,
\end{align*}
where the second inequality holds because the interim allocations are bounded within $[0,1]$,
and the last inequality holds by the continuity assumption (\cref{asp:continuity}).

Moreover, note that the utility $\tilde\util_{\alpha,z}$ increases at a rate of at most $\ability$ after type $\tilde\type^{(1)}$, 
while the utility $\hat\util_{\alpha,z}$ increases at a rate of $\ability$ within the interval $(\tilde\type^{(1)},\hat\type^{(3)})$. 
Therefore, the gain in utility is at least 
\begin{align*}
    & (1-\alpha) \cdot \int_{\tilde\type^{(1)}}^{\hat\type^{(3)}}\hat{\util}_{\alpha,z}(\type)\dd\dist(\type)
    - (1-\alpha) \cdot \int_{\tilde\type^{(1)}}^{\hat\type^{(3)}} \tilde\util_{\alpha,z}(\type)\dd\dist(\type)\\
    \geq\,& (1-\alpha) \cdot (\dist(\hat\type^{(3)}) - \dist(\tilde\type^{(1)})) 
    \cdot (\hat{\util}_{\alpha,z}(\tilde\type^{(1)}) - \tilde\util_{\alpha,z}(\tilde\type^{(1)})) \\
    \geq\,& (1-\alpha) \cdot (\dist(\hat\type^{(3)}) - \dist(\tilde\type^{(1)})) 
    \cdot (\ability\cdot (\tilde\type^{(1)} - \hat\type^{(1)})- \epsilon_1) \\
    \geq\,& \frac{1}{2\ability}(1-\alpha) \cdot \underline{\beta}_1 \cdot (\ability\cdot (\tilde\type^{(1)} - \hat\type^{(1)})- \epsilon_1).
\end{align*}
Since the matching efficiency in the contest $(\tilde{\alloc}_{\alpha,z},\tilde{\util}_{\alpha,z})$ is bounded above by the efficient allocation rule, 
combining the inequalities, we have that 
\begin{align*}
&\text{Obj}_{\alpha}(\tilde{\alloc}_{\alpha,z},\tilde{\util}_{\alpha,z}) - \text{Obj}_{\alpha}(\hat{\alloc}_{\alpha,z},\hat{\util}_{\alpha,z}) \\
&\leq \alpha\bar{\beta}_1 \cdot(\hat\type^{(2)}-\hat\type^{(1)})^2
- \frac{1}{2\ability}(1-\alpha) \cdot \underline{\beta}_1 \cdot (\ability\cdot (\tilde\type^{(1)} - \hat\type^{(1)})- \epsilon_1)\\
&\leq \alpha\bar{\beta}_1 \cdot\rbr{(\epsilon_2+\epsilon_0)^2\cdot\frac{\bar{\beta}_1}{\underline{\beta}_1}+\epsilon_0+\epsilon_2}^2
- \frac{1}{2\ability}(1-\alpha) \cdot \rbr{\ability\rbr{\epsilon_2 - \epsilon_0- \sqrt{\frac{8\epsilon_0\bar{\beta}_1}{\ability\underline{\beta}_1}}} - \epsilon_1}<0.
\end{align*}
The last inequality comes from the choice of $\epsilon_0, \epsilon_1, \epsilon_2$. 
Therefore, 
the contest $(\tilde{\alloc}_{\alpha,z},\tilde{\util}_{\alpha,z})$ is not optimal. 
\end{proof}

\end{document}